\documentclass[acmsmall,10pt]{acmart}

\AtBeginDocument{%
  }

\usepackage{thmtools}
\usepackage{thm-restate}
\usepackage{graphicx} 
\usepackage[algosection,ruled,lined,linesnumbered,longend]{algorithm2e}
\usepackage{algorithmic}
\usepackage{latexsym}
\usepackage{graphicx,wrapfig,xcolor}
\usepackage{float}
\usepackage{caption}
\usepackage{subcaption}
\usepackage{comment}
\usepackage{xcolor}
\usepackage{algorithm2e}
\usepackage{wrapfig}

\usepackage{todonotes}

\newcommand{\beq}{\begin{equation}}
\newcommand{\bet}{\begin{table}}
\newcommand{\eeq}{\end{equation}}

\newcommand{\N}{\mathbb{N}}

\newcommand{\opt}{\textsc{opt}\xspace}
\newcommand{\onl}{\textsc{onl}\xspace}

\newcommand{\work}{\textsc{work}}

\newcommand{\E}{\mathbb{E}}

\newtheorem{theorem}{Theorem}[section]
\newtheorem{lemma}[theorem]{Lemma}
\newtheorem{definition}[theorem]{Definition}
\newtheorem{corollary}[theorem]{Corollary}

\newtheorem{prop}[theorem]{Proposition}

\newtheorem{rem}[theorem]{Remark}
\newtheorem{remark}[theorem]{Remark}

\newtheorem{property}[theorem]{Property}

\newcommand{\defn}[1]{{\textit{\textbf{\boldmath #1}}}}

\usepackage{hyperref}
\usepackage{MnSymbol} 
\newcommand{\poly}{poly}

\newcommand{\rathish}[1]{}
\newcommand{\hao}[1]{}
\newcommand{\haonew}[1]{}
\newcommand{\rathishnew}[1]{}
\newcommand\myworries[1]{}

\title{Almost Tight Approximation Hardness and Online Algorithms for Resource Scheduling}
\author{Rathish Das}
\affiliation{%
 \institution{University of Houston}
 \city{Houston}
 \state{Texas}
 \country{United States}}
 \email{rathish@central.uh.edu}

\author{Hao Sun}
\affiliation{%
 \institution{University of Houston}
 \city{Houston}
 \state{Texas}
 \country{United States}}
 \email{hsun33@central.uh.edu}

\setcopyright{rightsretained}

\acmDOI{} 

\ccsdesc[500]{Theory of computation~Scheduling algorithms}
\ccsdesc[500]{Theory of computation~Online algorithms}
\ccsdesc[500]{Theory of computation~Shared memory algorithms}
\ccsdesc[500]{Theory of computation~Problems, reductions and completeness}

\keywords{Resource scheduling, precedence-constrained scheduling, hardness of approximation, online competitive resource scheduling algorithms}

\begin{document}

\begin{abstract}
We study the precedence-constrained resource scheduling problem [SICOMP'75]~\footnote{Part of this work has been presented in SPAA 2025~\cite{das2025approximation}.}. There are $n$ jobs where each job takes a certain time to finish and has a resource requirement throughout the execution time. There are precedence among the jobs. The problem asks that given a resource budget, schedule the jobs obeying the precedence constraints to minimize makespan (maximum completion time of a job) such that at any point in time, the total resource being used by all the jobs is at most the given resource budget. In the offline setting, an important open question is whether a polynomial-time $O(1)$-factor approximation algorithm can be found. We prove that hardness of constant factor approximation unless P=NP. 
 
In fact, we prove a stronger lower bound that for some constant $\alpha > 0$, there is no $o((\log t_{\max})^{\alpha})$-factor approximation algorithm for the offline version of the problem with maximum job length $t_{\max}$, even if there is only one type of resource, unless P = NP.
Further, we show that for some constant $\alpha > 0$, there is no $o( (  \log n )^\alpha )  $-factor approximation algorithm for the offline version of the problem, even if there is only one type of resource, unless NP $\subset$ DTIME$( O( 2^{    \text{polylog} (n)  }  )  ) $.

We further show a connection between this scheduling problem and a seemingly unrelated problem called the shortest common super-sequence (SCS) problem, which has wide application in Biology and Genomics. We prove that an $o(\log t_{\max})$-factor approximation of the scheduling problem would imply the existence of an $o(|\Sigma|)$-approximation algorithm for SCS with alphabet $\Sigma$.

We then consider the online precedence-constrained resource scheduling problem~\footnote{Independently and concurrently, Perotin, Sun, and Raghavan~\cite{perotin2025new} also came up with online algorithms for this problem.}. We present $\Omega(\log n)$ and $\Omega(\log t_{\max})$ lower bounds of the competitive ratio of any randomized online algorithm, where $n$ is the number of jobs and $t_{\max}$ is the maximum job length. Moreover, we present a matching $O(\min\{\log n, \log t_{\max}\})$-competitive deterministic online algorithm. 
\end{abstract}

\maketitle

\section{Introduction}

Precedence-constrained resource scheduling~\cite{garey1975bounds,  svensson2010conditional, augustine2006strip,   queyranne2006approximation, hall1997scheduling, MalleableTree, Malleable2006, perotin2021moldable, precedenceprocessor, citation-key, BRUCKER19993, KOLISCH2001249, Malleable2013} 
\hao{ references \cite{niemeier2015scheduling, jansen2019approximation, Mounie1999Malleble,jansen2004Malleble, lepere2002approximation} do not use precedence constraints  \cite{perotin2021moldable} moldable scheduling \cite{precedenceprocessor}  simply has multiple processors, \cite{citation-key} is a thesis, \cite{KOLISCH2001249} does have resource and precedent? but has other things as well } 
is an important problem and has been studied for several decades. In this problem, we are given a set of jobs, and there are precedence constraints among the jobs: if job $i$ has precedence over job $j$, then $j$ can only start after job $i$ is finished. Besides precedence constraints,  each job has a resource requirement (i.e., the job needs one or more than one unit of a resource during its execution). There is a total resource budget constraint. The goal is to finish the execution of all the jobs as soon as possible such that the total resource usage at any point in time does not exceed the given resource budget.



The precedence-constrained resource scheduling problem is  NP-hard~\cite{ScheduleNPhard} in the offline setting.
There are known $O(\log n)$-approximation algorithms~\cite{augustine2006strip} for the offline precedence-constrained resource scheduling problem where $n$ is the number of jobs. It remains an important open question of whether the approximation ratio can be improved. 


\paragraph{\textbf{This paper.}} We prove the hardness of approximation (which almost matches the current known upper bound in the approximation ratio) of the precedence-constrained resource scheduling problem posed nearly 50 years ago in~\cite{garey1975bounds}. Second, we consider the problem in the online precedence-constrained setting~\cite{graham1966bounds, azar2002line, DBLP:conf/stoc/FeldmannKST93, megow2009stochastic, DBLP:conf/icml/LassotaLMS23, epstein1998lower, epstein2000note} where a job is revealed after all of its predecessor jobs have been completed. We present matching upper and lower bounds in the competitive ratio in the online setting.

\paragraph{\textbf{Problem definition.}} We redefine the \textit{precedence-constrained resource scheduling problem} from ~\cite{garey1975bounds} as follows.
There are $d \ge 1$ types of resources $\mathcal{R}_1, \mathcal{R}_2 \ldots, \mathcal{R}_d$, where resource $\mathcal{R}_i$ has a budget of $r_i$ units. Examples of resource types are machines, storage space, or network bandwidth~\footnote{E.g., job $j$ may require 1 machine and 4MB memory to execute, while job $j^{'}$ may require 8 machines and 32MB memory to execute for the entire time of their executions.}. Resources are reusable, i.e., after a job relinquishes its resource (e.g., the machine on which the job is running), the resource can be reused by other jobs.

Let $J$ be a set of jobs where each job $j\in J$ takes $t_j$ units of time to complete and requires $r_{i,j}\in [0, r_{i}]$ units of resource $\mathcal{R}_i$ for $i \in [1,d]$ at all times during its execution,

There are precedence constraints given by a partial order $\prec$ on the jobs such that if $j \prec j^{'}$, then job $j$ must complete before $j^{'}$ can start (i.e., $j$ is a predecessor of $j^{'}$). This can be modelled as a directed acyclic graph (DAG) $D = (V, E)$, where the vertices $V$ of the DAG are the jobs and if there is a directed edge $(u,v)\in E$ from vertex $u$ to $v$, then $u \prec v$. 

A \defn{feasible schedule} is an assignment of the jobs to their starting times such that at each point in time, the total resource $\mathcal{R}_i$ usage is at most $r_i$ for $i\in[1,d]$ and no job starts before each of its predecessors completes.

In the offline setting of the problem, the whole input DAG is given in the beginning.
In the \textbf{online setting}, as defined in~\cite{graham1966bounds, azar2002line, DBLP:conf/stoc/FeldmannKST93, megow2009stochastic, DBLP:conf/icml/LassotaLMS23, epstein1998lower, epstein2000note}, 
the jobs are unknown beforehand, and a job is revealed after all of its predecessor jobs have finished. 

In both the online and offline versions of the problem, the goal is the same: Given a resource budget $r_i$ of resource $\mathcal{R}_i$ for $i \in [1,d]$ and given a DAG $D$, generate a feasible schedule with minimum makespan (time taken to complete all the jobs).


Throughout this paper, we consider \defn{non-preemptive scheduling}, that is, once a job starts executing, it cannot be stopped until it is finished. \myworries{On reviewer said we should talk about the nonpreemptive case.  It should be known that preemptive is factor 2 approximable right?   \cite{azar2002line} talk about preemptive but its with processors no resource \cite{hall1997scheduling}  minimize avg completion time 3 apx   fairly sure makespan is also easy to minimize   \cite{schedule2NPhard}  preemptive NPhard \cite{epstein1998lower} lower bound including preemptive  }
\myworries{Is the way worded in section 6 more formal and concise? Do we need both an introductory definition and a short formal one? }\myworries{mention preemptive easy apx?}
\hao{Is the way worded in section 6 more formal and concise? Do we need both an introductory definition and a short formal one?  \newline
Here we are given a DAG $D = (V,E)$ and $d \ge 1$ types of resource where resource $\mathcal{R}_i$ has capacity $r_i$ (after normalizing separately for each resource type, assume that each resource type has capacity $r$).
Each node $v\in V$  is a job that requires $r_{i,v}$ units of resource $\mathcal{R}_i$ for $t_v$  units of time to complete.   
For each edge $(u,v) \in E$,  we must complete job $u$ before starting job $v$. 
At any time, for each resource $\mathcal{R}_i$, we may use a total of $r$ units of resource $\mathcal{R}_i$.  
The objective is to complete all jobs as soon as possible. }

Resource scheduling problems in parallel computation are critically important due to their direct impact on the performance, efficiency, and overall effectiveness of parallel computing systems. In any parallel computing environment, resources such as processors, memory, and network bandwidth are limited. By minimizing waiting times and delays, effective resource scheduling helps reduce the makespan---the total time required to complete all tasks/jobs. The faster tasks are completed, the more efficient the parallel system becomes, resulting in quicker outcomes and increased productivity.

\paragraph{\textbf{Motivation of the online setting.}} 
The online model of precedence-constrained scheduling is very well-motivated  \myworries{well motivated -$>$ well-motivated ?} and studied for many decades. Some of the practical applications of this model where a job is revealed after all of its predecessors finish execution are as follows. 
\defn{Multi-threaded parallel computation}~\cite{blumofe1993space, cormen2022introduction, blumofe1999scheduling} can be modelled as a DAG: Each node $v$ in the DAG is a thread $\mathcal{T}_u$ (a sequence of instructions), and an edge from node $u$ to node $v$ represents thread $\mathcal{T}_u$ either spawns or returns to thread $\mathcal{T}_v$. The DAG \defn{unfolds dynamically}, i.e., the runtime scheduler does not know about a thread until all of its predecessor threads complete their execution. There are many implementations of online multi-threaded schedulers, such as Cilk~\cite{blumofe1995cilk, frigo1998implementation, bender2000scheduling}, that are built on the model that a job is revealed after all of its predecessors are finished.   

Another important application of the online model of precedence-constrained scheduling is \defn{Cloud computing}~\cite{DBLP:conf/icml/LassotaLMS23, voorsluys2011introduction, arunarani2019task} services offer users on-demand \myworries{on-demand ?} execution of complex tasks, such as scientific computations. Such tasks can be decomposed into smaller jobs where a job is dependent on the data produced by other smaller jobs, user input, and other dynamics of the system. Hence, the processing times and inter-dependencies (precedence constraints) of such smaller jobs become known only after certain jobs are finished and their results/output can be evaluated. 

All these online schedulers, be in cloud computing or in multi-threaded computation, need to manage one or more than one  \myworries{ 
one or more than one -$>$  one or more? } types of resources for the tasks they schedule. Threads in a multi-threaded computation might require resources such as processing cores, cache space, DRAM space etc~\cite{schneider2006scalable, berger2000hoard}. Tasks in a cloud computing system might require super-computing nodes, GPU nodes, network bandwidth, storage space, etc~\cite{diab2013dynamic, zhu2011multimedia}. Each of these resources is not unlimited, instead, there is a fixed budget for each type of resource.

\paragraph{\textbf{Our contribution.}} 

\begin{itemize}
\item \textbf{(Offline setting.) [Section~\ref{sec:lts}]} We show the impossibility of any polynomial-time $O(1)$-factor approximation algorithm for the precedence-constrained resource scheduling problem, unless P=NP. In particular, we prove a stronger lower bound that for some constant $\alpha > 0$, there is no $o((\log t_{\max})^{\alpha})$-factor approximation algorithm for the offline version of the problem with maximum job length $t_{\max}$, \defn{even if there is only one type of resource}, unless P = NP. 

\begin{restatable*}{theorem}{ResourceImplyLTS}
\label{thm:lts-new}
    Let $t_{\max}$ be the maximum job length. For some $\alpha >0$, there is no polynomial-time $o( (  \log t_{\max} )^\alpha )  $-factor approximation for the precedence-constrained resource scheduling problem, even if there is only one type of resource, unless P = NP. 
\end{restatable*}

We then present another hardness of approximation based on the number of jobs.

\begin{restatable*}{theorem}{ResourceImplyLTSJobs}
\label{thm:lts-jobs}      Let $n= |V|$ be the number of jobs of our resource scheduling instance. For some $\alpha >0$, there is no polynomial-time $o( (  \log n )^\alpha )  $-factor approximation for the precedence-constrained resource scheduling problem, even if there is only one type of resource, unless  NP $\subset$ DTIME$( O( 2^{    \text{polylog} (n)  }  )  ) $.
  \end{restatable*}

 \item \textbf{(Offline setting.) [Section~\ref{sec:scs}]} While there is still a gap between the $O(\log t_{\max})$ upper bound~\footnote{In this paper (Section~\ref{OnlineAlgSection}), we present an $O(\log t_{\max})$-competitive online algorithm (which also immediately gives an offline algorithm with $O(\log t_{\max})$-approximation ratio) for this problem.} and $o((\log t_{\max})^{\alpha})$ lower bound for some $\alpha >0$ in the approximation ratio of the offline resource scheduling problem, we prove that improving the approximation ratio from $O(\log t_{\max})$ in the offline resource scheduling problem would require improving the current state of the art of a long-standing open problem called \emph{shortest common super-sequence problem} (SCS)\cite{SCSBinNPHard,supseqnotAPX,SCS2(3)APX}.  \myworries{is there a nice way to clarify $ \alpha >0 $ in \cite{LTSnotAPX}  }

 \begin{restatable*}{theorem}{ResourceImplySCS}
\label{thm:scs}
    Even if there is only one resource type, for any fixed $\rho \in \mathbb{N}$, if there exists a polynomial-time $o(\log t_{\max}) $-approximation algorithm for the resource scheduling problem with maximum job size $t_{\max} = 2^{\rho}$, then there exists a polynomial-time $o(\rho)$-approximation algorithm for SCS on an alphabet of size $\rho$.
\end{restatable*}

\item \textbf{(Online setting.) [Section~\ref{sec:lower-online}]} We present a lower bound on the competitive ratio of any algorithm (even randomized ones) for the online precedence-constrained resource scheduling problem. 


\begin{restatable*}{theorem}{OnlineRSHard}\label{Thm:online-lower}
Even if there is only one type of resource, no randomized online algorithm for our resource scheduling problem can be better than $o(\log n)$-competitive or $o(\log t_{\max} )$-competitive where $n$ is the number of jobs in the input DAG and $t_{\max}$ is the maximum job length.
\end{restatable*}

\begin{restatable*}{theorem}{thmonlinelowermult}\label{thm:online-lower-multi}
 Let $d\ge 3$ be the number of resource types. Then, no algorithm (even randomized ones) for online multi-resource scheduling can achieve a competitive ratio better than $\frac{d-1}{2}$. 
\end{restatable*}

\item \textbf{(Online setting.) [Section~\ref{OnlineAlgSection}]} We then also present a deterministic online algorithm with a competitive ratio that matches the lower bound---thus, we achieve a tight competitive ratio for this problem. 
\begin{restatable*}{theorem}{OnlineLogApx}\label{thm:onl-upper}
Let $V$ be the set of jobs, $d \ge 1$ be the number of resource types, and $t_{\max}$ be the maximum processing time of any job. Then there is a deterministic online algorithm with competitive ratio $O(d + \min\{\log |V|, \log t_{\max}\})$ for the online precedence-constrained resource scheduling problem.    \myworries{the problem or problems?}
\end{restatable*}

\end{itemize}

\textit{A key technical contribution of this paper is the introduction of a tool, called \defn{chains} [Section~\ref{sec:chains}] that we use heavily to prove all of our lower bounds in this paper (lower bound in competitive ratio in the online setting and hardness of approximation in the offline setting). We hope that this new tool, chains will find its usefulness further in future work on resource scheduling.
}

\paragraph{\textbf{Relating scheduling with shortest common super-sequence.}}
We establish a connection between precedence-constrained resource scheduling problem and a seemingly unrelated problem, called the Shortest Common Super-sequence (SCS) problem. \myworries{  -$>$ and the seemingly unrelated Shortest Common Supersequence (SCS) problem? }
SCS has a long history of work for many decades and has wide applications in Biology, Geonomics, etc~\cite{supseqnotAPX, SCSExperimental}. 
In the SCS problem,  we are given a list $L$ of sequences of integers from an alphabet $\Sigma = {1,2,3,\ldots,\rho}$ for some $\rho$ and wish to find a sequence  of minimal length that is a supersequence for each sequence of $L$. 
SCS is known to be NP-hard even on a binary alphabet \cite{SCSBinNPHard}.
For a fixed-size \myworries{fixed-size?} alphabet $\Sigma$, Jiang and Li~\cite{supseqnotAPX} proved an $|\Sigma|$-approximation algorithm for SCS. 

We show that getting a polynomial-time $o(\log t_{\max})$-factor approximation algorithm for the scheduling problem implies a polynomial-time $o(|\Sigma|)$-approximation algorithm for SCS, where $t_{\max}$ is the maximum job length. \myworries{ scheduling problem -$>$ the scheduling problem? }




\paragraph{\textbf{ Summary of our results.}}
We summarize our results as follows.
\begin{enumerate}
\item (Offline setting.) For some constant $\alpha > 0$, there is no $o((\log t_{\max})^{\alpha})$-factor approximation algorithm 
\hao{actually is $\frac{1}{8} (\log t_{\max})^{\alpha})  $ include or no?} for the offline version of the problem, even if there is only one type of resource (i.e., $d = 1$), unless P=NP (Theorem~\ref{thm:lts-new}).

\item (Offline setting.) For some constant $\alpha > 0$, there is no $o( (  \log |V| )^\alpha )  $-factor approximation algorithm  for the offline version of the problem, even if there is only one type of resource (i.e., $d = 1$), unless NP $\subset$ DTIME$( O( 2^{    \text{polylog} (n)  }  )  ) $
(Theorem~\ref{thm:lts-jobs}).

\item (Offline setting.) If there exists a polynomial-time $o(\log(t_{\max}))$-approximation algorithm for our scheduling problem (even with only one resource type),  then there exists a polynomial-time $o(|\Sigma|)$-approximation algorithm for the SCS problem (Theorem~\ref{thm:scs}).

\item (Online setting.) We present a deterministic $O(d  + \min\{\log |V|, \log t_{\max}\})$-competitive online algorithm for the online precedence-constrained resource scheduling problem, where $V$ is the set of jobs, $t_{\max}$ is the maximum processing time of a job, and $d$ is the number of resource types (Theorem~\ref{thm:onl-upper}). 

\item (Online setting.) No online algorithm, even with randomization, can be $o(\min\{\log |V|, \log t_{\max}\})$-competitive, even if there is only one type of resource (i.e., $d = 1$) (Theorem~\ref{Thm:online-lower}). 

\item (Online setting.) No online algorithm, even with randomization, can be $\frac{d-1}{2}$-competitive \hao{probably less clunky than saying there is no f(d) f...} when there are $d\geq 3$ types of resources, even if $d$ is a fixed constant (Theorem~\ref{thm:online-lower-multi}).

\end{enumerate}

\rathishnew{We present a technical overview of the paper in the supplementary material.}
\subsection{Technical Overview}

We now present an essential toolbox to prove the lower bounds of our problem. This toolbox, chains, is applicable even if there is only one resource type. \hao{mention only need single resource for lower bound?}

\paragraph{\textbf{ A new toolbox for resource scheduling---chains.}}
Call a job \defn{skinny} (in the single resource case\haonew{do we need to specify this?}) if its resource requirement is $\epsilon < 1/n$ (i.e., very small) where $n$ is the number of jobs. Hence, all the skinny jobs can be run in parallel (not restricted due to the resource constraint). Call a job \defn{fat} if its resource requirement is 1 (i.e., takes all the resource, assuming the resource budget is normalized to 1). 

Call a pair of jobs, consisting of a skinny job of length $2^i$ followed by a fat job of length 0, a \defn{tuple-$i$} (i.e., the skinny job has precedence over the fat job). Without loss of generality, we can allow jobs of 0 length as discussed at the end of the Introduction.

A chain $C$ of type \defn{$C(m,i)$} is created by concatenating $2^m/2^i$ number of tuple-$i$ (here, $m \ge i$): the fat job of the first tuple-$i$ has precedence over the skinny job of the second tuple-$i$, the fat job of the second tuple-$i$ has precedence over the skinny job of the third tuple-$i$, and so on. The length of chain $C$ is $2^i \cdot 2^m/2^i = 2^m$ \textbf{(see Figure~\ref{fig:chains}(a))}.

We call the node of a chain with no in-neighbours (i.e., the skinny job in the first tuple) the \emph{source} and the node with no out-neighbours (i.e., the fat job in the last tuple) the \emph{sink}.  \myworries{i.e, -> i.e.,}

The precedence among the chains is defined as follows: If chain $C_1 \prec C_2$, then all jobs in chain $C_1$ must be finished before a job from chain $C_2$ starts, i.e., the sink node of chain $C_1$ has precedence over the source node of chain $C_2$ \textbf{(see Figure~\ref{fig:chains}(b))}.

An important property of chains: If we run $p$ chains of the same type, each of length $x$, then the makespan is $O(x)$, however, if the chains are of different types, the makespan is $\Omega(p\cdot x)$.

The property makes chains of different types incompatible to each other. \myworries{The property -> This property?  incompatible to -> incompatible with? } We use this property to create gadgets to prove all of our lower bounds in this paper (lower bound in competitive ratio in the online setting and hardness of approximation in the offline setting). We hope that this new tool, chains will find its usefulness further in future work on resource scheduling.

\rathish{Reviewer 2 said that the hardness of approximations use clever gadget constructions that make different chain of jobs incompatible. 
Technical takeaway: We hope, this tool---chains will also be useful in future resource scheduling work.}

\paragraph{\textbf{ Hardness of $o( ( \log t_{max} )^\alpha )$-factor approximation for the offline problem.}}
We do a hardness reduction from the \defn{loading-time scheduling problem (LTS)}~\cite{LTSnotAPX}. In the LTS problem, there is a DAG $D = (V,E)$ where $V$ represents the set of jobs and directed edges in $E$ represent the precedence relations among the jobs. There is a set of machines $m$ with their loading times $\ell(m)$. Each job in $V$ can be run on a specific machine. A feasible LTS solution is a partition of the jobs into sets $X_1, X_2, \ldots X_p$ such that all the jobs in each $X_i$ can be run on the same machine $m_i$ and if job $j \prec j^{'}$ and $j$ is put in partition $X_i$, then $j^{'}$ must be put in partition $X_{i^{'}}$, where $i \le i^{'}$. The cost of a partition $X_i$ is $\ell(m_i)$, the loading time of machine $m_i$. The cost of the LTS solution is the sum of the cost of the partitions $X_i$ for $i\in [1, p]$. The goal is to find a feasible solution with minimum cost. A detailed definition of LTS can be found in Section~\ref{sec:lts-revisit}.

From~\cite{LTSnotAPX}, there is a hardness of approximation for the LTS problem as follows. \myworries{LTS -> the LTS}
\rathish{Hao, could you please put LTS's $\log$ hardness theorem as restateble here?}
 \begin{restatable*}{theorem}{LTSorig}\label{LTSoriginal}
 \cite{LTSnotAPX}  Let $\rho$ be the number of machines in the LTS problem. For some constant $\alpha >0$,    no polynomial-time algorithm can achieve a $\rho^ \alpha$-approximation ratio  for the LTS problem in the restricted case, where the number of machines $\rho$ is a fixed constant and each job can be performed only on a single machine, unless P= NP.
\end{restatable*}

The LTS problem cannot readily be used as its current definition for the hardness reduction---Consider the situation where there are three jobs $j_1, j_2$ and $j_3$. Let all three jobs can be assigned on the same machine $m$ and $j_1\prec j_2 \prec j_3$. Then an LTS solution, where all three jobs are put in the same partition, costs only the loading time of machine $m$ once (processing times of jobs are not taken into cost calculation). However, in the scheduling solution, if we run multiple jobs in a single machine, we need to run them sequentially, and thus their processing times add up.   

To resolve this issue, we transform the input graph of LTS into another graph  called a \defn{conflict-free} graph with lesser precedence relations, which bypasses this issue: In a conflict-free graph, for any two jobs $j$ and $j^{'}$ with the same machine assignment $m$, if $j \prec j^{'}$, then there must be another job $j^{''}$ with a different machine assignment $m^{''}$ and $j \prec j^{''} \prec j^{'}$. In this case, any feasible LTS solution must put $j, j^{'}$ and $j^{''}$ in three different partitions. This allows us to relate the sum of processing times of $j, j^{'}$  and $j^{''}$ in the resource scheduling problem to the sum of loading costs of the three partitions corresponding to $j, j^{'}$  and $j^{''}$ in the LTS solution.
We then prove that the inapproximability results of LTS hold in the reduced intermediate LTS problem with conflict-free graphs.

We still need to do a further transformation of the intermediate LTS with a conflict-free graph since the loading times of the machines may not be bounded and we want the loading times to be upper bounded by a polynomial in the problem size, i.e., polynomial in the number of jobs and machines. We reduce the intermediate problem to a conflict-free \defn{bounded} LTS problem and show the same asymptotic inapproximability guarantee still holds.

We then use our tool, \defn{chains} to reduce the conflict-free bounded LTS problem to our resource scheduling problem and prove the following theorem.

\begin{restatable*}{theorem}{ResourceImplyLTS}
\label{thm:lts-new}
    Let $t_{\max}$ be the maximum job length. For some $\alpha >0$, there is no polynomial-time $o( (  \log t_{\max} )^\alpha )  $-factor approximation for the precedence-constrained resource scheduling problem, unless P = NP. 
\end{restatable*}
\paragraph{\textbf{ Hardness of $o( (  \log |V| )^\alpha )  $-factor approximation for the offline problem.}}
The approach is similar to the hardness of $o( ( \log t_{max} )^\alpha )$-factor approximation.  The key difference is that the number of machines is no longer a constant, so we must look into the details of the arguments in \cite{LTSnotAPX} in regards to how the number of machines relates to the number of jobs.
We show the following theorem.
\begin{restatable*}{theorem}{ResourceImplyLTSJobs}
\label{thm:lts-jobs}      Let $|V|$ be the number of jobs of our resource scheduling instance. For some $\alpha >0$, there is no polynomial-time $o( (  \log |V| )^\alpha )  $-factor approximation for the precedence-constrained resource scheduling problem, even if there is only one type of resource, unless $NP \subset DTIME( O( 2^{    \text{polylog} (n)  }  )  ) $.
  \end{restatable*}

\paragraph{\textbf{ An $o(\log t_{\max})$-approximation in scheduling implies improved approximation for SCS.}}
We make use of our lower bound construction tool, \defn{chains} to prove the conditional hardness result.
Given an SCS instance with alphabet $\{1,2,3,\ldots, \rho\}$ for some constant $\rho$, we create an instance of our scheduling problem with resource and precedence constraints by corresponding each character $i$ in an input sequence with a chain $C(m,i)$. By doing an approximation preserving reduction, we prove the following theorem.

\begin{restatable*}{theorem}{ResourceImplySCS}
\label{thm:scs}
    For any fixed $\rho \in \mathbb{N}$, if there exists a polynomial-time $o(\log t_{\max}) $-approximation algorithm for the resource scheduling problem with maximum job size $t_{\max} = 2^{\rho}$, then there exists a polynomial-time $o(\rho)$-approximation algorithm for SCS on an alphabet of size $\rho$.
\end{restatable*}
The current best polytime approximation ratio for SCS is $\rho$, even if the alphabet size $\rho$ is a fixed constant~\cite{supseqnotAPX}. \hao{added reference subseqnotAPX has this as a theorem without citing anyone else but it's so simple} Finding a polynomial-time $o(\rho)$-approximation algorithm for SCS is a long standing open problem.

\paragraph{\textbf{ Lower bound construction for online algorithms.}}
Using our toolbox, Chains, we present a lower bound in the competitive ratio of any online algorithm (even randomized ones). We create a random input instance which is a gadget of $\log |V|$ different types of chains. There are precedence relations among the chains. The high-level idea is that the offline algorithm knows the future, and thus can group chains of the same type (from our chain toolbox, we show that same type chains are compatible with each other) and run them in parallel. On the other hand, the online algorithm cannot put chains of the same type together and becomes $\Omega(\log |V|)$-factor worse in makespan than the offline optimal algorithm.
 
To prove a $\Omega(\log |V|)$-competitive and $\Omega(\log \ell)$-competitive lower bounds simultaneously, we create a single randomized DAG instance with maximum job length $\ell$ where $\log |V| = \Theta(\log \ell)$.  We thus prove the following theorem.

\begin{restatable*}{theorem}{OnlineRSHard}\label{Thm:online-lower}
No randomized online algorithm for our resource scheduling problem can be better than $o(\log n)$-competitive or $o(\log \ell)$-competitive where $n$ is the number of jobs in the input DAG and $\ell$ is the maximum job length.
\end{restatable*}

We also prove the following lower bound (even for randomized algorithms) as a function of the number of resource types.

\rathish{Hao, could you please put the theorem from section 4.2 as restatable here?}
\begin{restatable*}{theorem}{thmonlinelowermult}\label{thm:online-lower-multi}
 Let $d$ be the number of resource types. Then, no algorithm (even randomized ones) for online multi-resource scheduling can achieve a competitive ratio better than $\frac{d-1}{2}$. 
\end{restatable*}

\paragraph{\textbf{ Designing an online algorithm---A greedy strategy could be arbitrarily bad.}}
It might seem like to minimize the makespan, we should not leave any under-utilized resource, that is, we should schedule jobs if enough resource is available. 
We call an algorithm $\mathcal{A}$, \textbf{greedy} if whenever adequate resources are available in any point in time, $\mathcal{A}$ schedules job(s) from the set of jobs that are ready to execute. 
However, we show that every greedy algorithm could be as bad as $\Omega(n)$-competitive, where $n$ is the number of jobs (see Lemma~\ref{greedy-lower}). The problem arises due to the existence of long and skinny jobs (these jobs use significantly fewer resources but run for a long time) along with short and fat jobs (these jobs use all the resources but run for a very short time). 


  \paragraph{\textbf{ Making bad events rarer.}} The greedy algorithm suffers because we start the long and skinny job immediately. The size of the bottleneck due to a long and skinny job of length $\ell$ is equal to its length. The larger the value $\ell$, the larger the bottleneck. Call starting a long job of length $\ell$ is a bad event of length $\ell$. 
  
  \textit{We will make the bad events with longer lengths rarer}. Without loss of generality, assume that $\ell$ is a power of $2$; otherwise, make it the next nearest power of $2$. We schedule jobs of length $\ell = 2^i$ at carefully chosen some later time, even if adequate resource is available to schedule the job earlier.
  In what follows, we prove that all the jobs of the same length $2^i$ cause, in total, a bottleneck of length equal to the depth $d$ of the DAG. 
  
  
To show that $O(\log |V|)$-competitiveness, we design a different small optimization problem based on the depth and the number of jobs. Solving this optimization problem, we present the following theorem.

\begin{restatable*}{theorem}{OnlineLogApx}\label{thm:onl-upper}
Let $V$ be the set of jobs, $d \ge 1$ be the number of resource types, and $t_{\max}$ be the maximum processing time of any job. Then there is a deterministic online algorithm with competitive ratio $O(d + \min\{\log |V|, \log t_{\max}\})$ for the online precedence-constrained resource scheduling problem.    \myworries{the problem or problems?}
\end{restatable*}

\paragraph{\textbf{ Jobs with 0-processing time can be allowed without loss of generality.}}
\label{par:zero-time}
We define the resource scheduling problem with zero-time jobs to be the resource scheduling problem where some jobs $v$ have $t_v=0$. Such a job $v$ can be completed without advancing timestep.
Given a resource scheduling instance $G$ and $ 0 < \epsilon <1$, we can construct the resource scheduling instance $G^\epsilon$ that has no zero-time jobs as follows.

$ G^\epsilon $ has the same jobs (vertices) and precedence relations (edges) as $G$. The length $t'_v$ of job $v$ in $G^\epsilon$, is defined by $t'_v= \frac{1}{\epsilon}|V(G)| t_v  $ if $t_v \neq 0$ and $t'_v=1$ if $t_v =0$.
Then given a solution to process $G$ in time $\tau$, we can find a solution to process $G^\epsilon$ in time $ \frac{1}{\epsilon}|V(G)|  \tau + |V(G)| $ and vice-versa.
Hence it is possible to approximate resource scheduling on $ G^\epsilon $ within $c$ if and only if it is possible to approximate resource scheduling on $G $ within $ c- O(\epsilon) $. 

Throughout the paper, we thus allow jobs with 0-processing time.

\subsection{Related Work}

\paragraph{\textbf{ Scheduling with precedence and resource constraints.}}
Scheduling with precedence constraints is a classic problem studied extensively through many decades. The makespan minimization problem with precedence constraints is NP-hard, and Graham's celebrated list scheduling algorithm~\cite{graham1966bounds} gives a $(2 - 1/m)$ approximation for $m$ machines. This list scheduling algorithm can be used as an online algorithm. Recently, Svensson~\cite{svensson2010conditional} proved that the problem is hard to approximate better than factor $(2-\epsilon)$ assuming results from Unique Game Conjecture (UGC). 
For the general objective function, weighted completion time, there are O(1)-factor approximation algorithms~\cite{hall1997scheduling, queyranne2006approximation}.


Scheduling with resource constraints is also a classic problem, with many results that have been known over many years. In the offline setting, Garey and Graham~\cite{garey1975bounds} present a $(3 - 3/m)$-approximation algorithm. The result was later improved by Niemeier and Wiese~\cite{niemeier2015scheduling} and, finally, an AFPTAS is known due to Jansen et al.~\cite{jansen2019approximation}.

Scheduling with resource constraints is closely related to the strip packing problem. In the strip packing problem, we are given a set of rectangles and a strip with a fixed width and unbounded height. The rectangles correspond to jobs where the width and height of a rectangle correspond to the resource requirements and duration of the job. The width of the strip corresponds to the total resource budget. The goal is to pack the rectangles into the strip, minimizing the heights that a rectangle can reach. Clearly, this minimizes the makespan of scheduling the jobs with resource constraints. There are many offline and online  strip packing algorithms~\cite{christensen2017approximation}. 

Baker and Schwarz~\cite{baker1983shelf} presented an asymptotic competitive ratio of 1.7 for the online strip packing problem. Hence, this result applies to online scheduling with resource constraints.   These results were further improved by subsequent research~\cite{csirik1997shelf, hurink2008online, ye2009note, han2007strip}.

Augustine et al.~\cite{augustine2006strip} gave the first approximation algorithm for scheduling with resource and precedence constraints in the offline setting. Given $n$ jobs, they achieve an $O(\log n)$-approximation algorithm. Demirci et al.~\cite{demirci2018approximation} work on this problem with an additional constraint on machine availability. They also achieve an $O(\log n)$-approximation algorithm. 

To the best of our knowledge, there is no non-trivial hardness of approximation ratio known for the offline setting when the number of resource types is a \defn{fixed constant}. On the other hand, if the number of resource types is not a constant, then the job-shop scheduling becomes a special case to the precedence-constrained resource scheduling problem~\footnote{Each job is a chain of tasks that must be done in order. Each task fully requires a specific resource.}. There is an $\Omega(\log^{1-\epsilon} \text{LB})$ NP-Hardness approximation lower bound where LB is a lower bound on optimal for job-shop scheduling due to Mastrolilli and Svensson~\cite{mastrolilli2008acyclic}.

Scheduling jobs with multiple resource types is closely related to the vector scheduling problem proposed by Chekuri and Khanna~\cite{chekuri2004multidimensional}. In this problem, each job has $d$-different types of resource requirements, which is modelled as a $d$-dimensional vector, and the goal is to  minimize the maximum load over all dimensions. Im et al.~\cite{im2015tight} worked on the online version of the vector scheduling problem, which is quite practical. They consider scheduling a set of independent jobs/vectors (no precedence constraints among the vectors), where jobs come online at some time, unknown at the beginning. They proved tight bounds for their problem. 
 
To the best of our knowledge, there is no non-trivial upper or lower bound on the competitive ratio (as a function of the number of jobs or maximum job sizes) known for the online precedence-constrained resource scheduling problem. 

There is another line of work where the resource requirements of jobs are malleable\cite{Ludwig1994SchedulingMallable,banerjee1990MallableApx,Leung1989Mallable, Mounie1999Malleble,  jansen2002Malleble, jansen2004Malleble, DasTsDu19}, that is the jobs can be executed with fewer resources at the expense of increasing their processing time. \hao{malleable scheduling is far more restrictive as it requires  that a c factor increase in resource usage decreases completion time by at least factor c should we emphasize that? } 
\hao{ references \cite{Ludwig1994SchedulingMallable, Mounie1999Malleble, jansen2002Malleble, jansen2004Malleble} do not use precedence constraints \cite{banerjee1990MallableApx} I can't see to find,  
same reference }
The approximation ratios achieved in malleable scheduling are $O(1)$, while for fixed resource requirement, the best known  approximation ratio is $O(\log n)$. 

Resource scheduling also has nice applications in parallel paging/caching problems~\cite{DBLP:conf/spaa/AgrawalBDKPS22, AgrawalBD20, DasAB20, AgrawalBD21, DBLP:conf/spaa/DeLayoZABBDMP22} where the cache size is treated as the resource budget.

\section{Chains---a New Toolbox for Resource Scheduling Lower Bounds}
\label{sec:chains}
In this section, we introduce a toolbox, called \defn{chains}, which will be very useful in proving the lower bound in the approximation ratio for the offline setting, as well as proving the lower bound in competitive ratio for the online setting.
~\paragraph{\textbf{ What is a chain?}}
Call a job \defn{skinny} if its resource requirement is $\epsilon < 1/n$ (i.e., very small) for all types of resource $\mathcal{R}_i$ for $i\in[1,d]$ where $n$ is the number of jobs. Hence, all the skinny jobs can be run in parallel (not restricted due to resource constraint). Call a job \defn{fat} if its resource requirement is 1 (i.e., takes all the resource, assuming resource budget is normalized to 1) for at least one resource type. 

Call a pair of jobs, consisting of a skinny job of length $2^i$ followed by a fat job of length 0, a \defn{tuple-$i$} (i.e., the skinny job has precedence over the fat job). Recall that without loss of generality, we can allow jobs of 0 size as discussed in the Introduction.

A chain $C$ of type \defn{$C(m,i)$} is created by concatenating $2^m/2^i$ number of tuple-$i$ (here, $m \ge i$): the fat job of the first tuple-$i$ has precedence over the skinny job of the second tuple-$i$, the fat job of the second tuple-$i$ has precedence over the skinny job of the third tuple-$i$, and so on. The length of chain $C$ is $2^i \cdot 2^m/2^i = 2^m$ (see Figure~\ref{fig:chains}(a)).

We call the node of a chain with no in-neighbours (i.e, the skinny job in the first tuple) the \emph{source} and the node with no out-neighbours (i.e, the fat job in the last tuple) the \emph{sink}. 

The precedence among the chains is defined as follows: If chain $C_1 \prec C_2$, then all jobs in chain $C_1$ must be finished before a job from chain $C_2$ starts, i.e., the sink node of chain $C_1$ has precedence over the source node of chain $C_2$ (see Figure~\ref{fig:chains}(b)).

\begin{figure*}[htb]
    \centering
      \includegraphics[width=\textwidth]{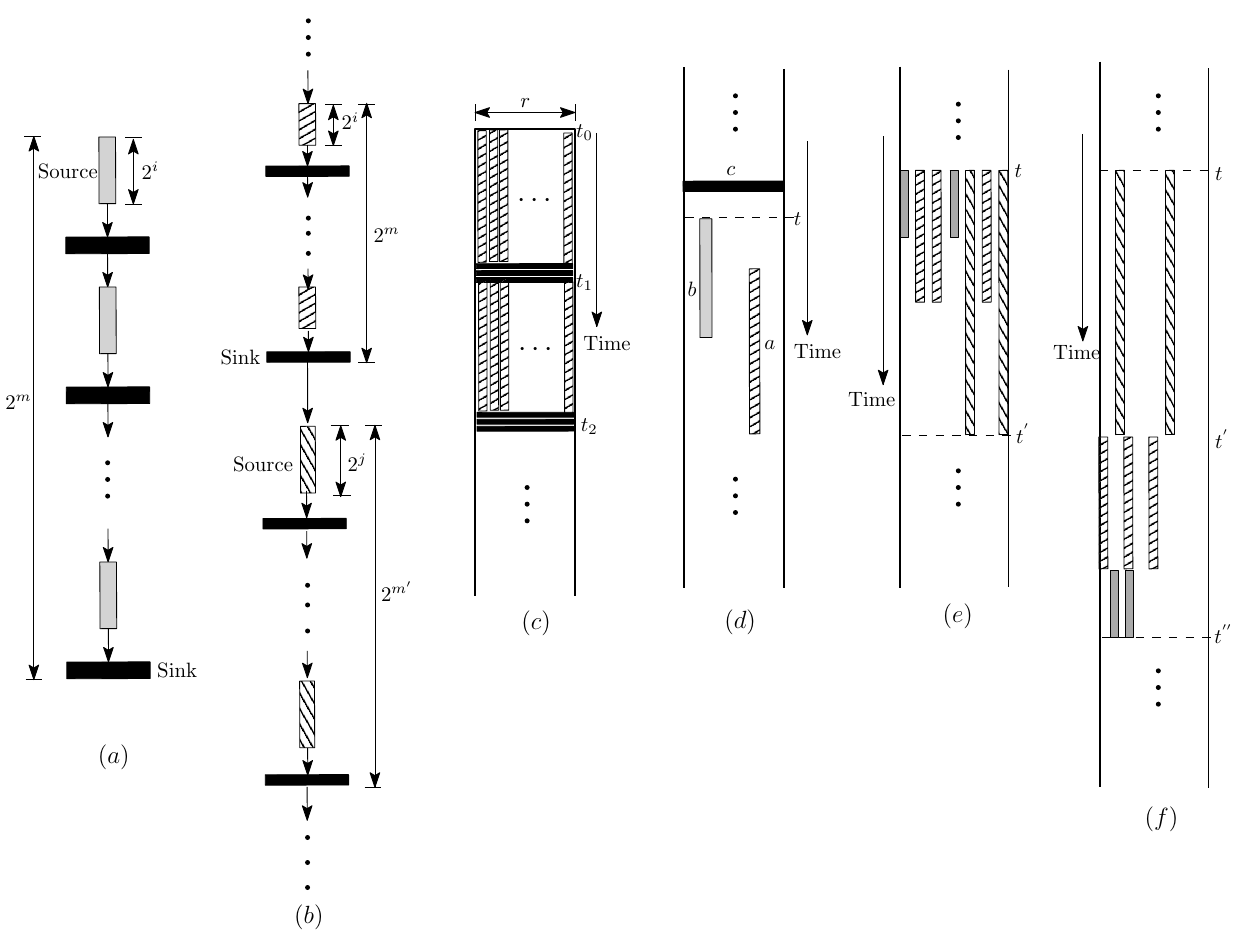}
\caption{Figure (a) shows a chain of type $C(m,i)$. \\
Figure (b) shows how two chains, $C_1$ of type $C(m,i)$ and $C_2$ of type $C(m^{'},j)$ are connected if $C_1 \prec C_2$: The sink node of $C_1$ has precedence over the source node of $C_2$.\\
Figure (c) shows the optimal way to schedule a set of chains whose skinny jobs have the same length.\\
Figure (d) shows that skinny job $a$ can be started at the same time when $b$ starts (for proof of Lemma~\ref{immediatestartchain}). \\
Figures (e) and $(f)$ show that by running only the skinny jobs of same length together, the makespan increases at most by a factor of 2 (for proof of Lemma~\ref{cleanschedule}).
}
    \label{fig:chains}
\end{figure*}


\subsection{Useful Properties of Chains}
Chains with skinny jobs of the same length can be run in parallel. On the contrary, chains with skinny jobs of different length cannot exploit parallelism: the fat jobs act as a barrier to start smaller length skinny jobs.   
\begin{property}\label{skinnyjobpar}
Suppose that there are $p$ independent chains (i.e., jobs from a chain has no precedence over jobs in another chain), of types $C(m_1,i), C(m_2, i), C(m_3, i), \ldots C(m_p, i)$, respectively (i.e., skinny jobs in all $p$ chains are of the same length). Then the optimum makespan of these $p$ chains is $\max\{2^{m_1}, 2^{m_2}, \ldots 2^{m_p}\}$.
\end{property}
\begin{proof}
    W.l.o.g let $ \max \{{m_1}, {m_2}, \ldots {m_p}\} =m_p $. Since the longest chain has length $2^{m_p}$, the makespan is at least  $2^{m_p}$.
    Now consider the following algorithm. For $j=1,.., 2^{m_p}/2^i $, run the $j$-th tuple of each unfinished chain as follows: In parallel, for each chain that yet has incomplete jobs, run the skinny job of the $j$-th tuple in the chain, followed by the fat job of the $j$-th tuple of each chain one after another. Each iteration $j$ takes a total time $2^i$ since skinny jobs of all $p$ chains can be run in parallel and the fat jobs have 0 length.
    
    At the $j$-th iteration for $j> m_t$, the chain $ C(m_t, i)$ will be completed. 
    So the entire graph is complete by the $2^{m_p}/2^i$ at which time the makespan will be $2^{m_p}$.
\end{proof}

\begin{property}\label{immediatestartchain}
    Let $\mathcal{C}$ is a set of chains. There could be precedence relations among the chains: if chain $C_1 \prec C_2$, then all jobs in chain $C_1$ must be finished before a job from chain $C_2$ starts. Given any algorithm $\mathcal{A}$ that processes chains $\mathcal{C}$ with makespan $T$, we can find an algorithm $\mathcal{A}'$ with makespan at most $T$ with the following property: algorithm $\mathcal{A}'$ does not start any job while another job has already started but not finished.
\end{property}

\begin{proof}
    Let $\mathcal{A}$ be any algorithm. 
    Suppose that some job $a$ starts while another job $b$ is running. Since $a$ and $b$ can run in parallel, both $a$ and $b$ must be skinny jobs (Figure~\ref{fig:chains}(d)). Since the predecessor of a skinny job, if exists, is a fat job, $a$'s predecessor (a fat job) must finish before $b$ started, since a fat job takes all the resource budget and thus can not be run in parallel with a skinny job.
    This means that we could have started $a$ when algorithm $\mathcal{A}$ starts job $b$. 
    Continue this process until there are no such jobs $a$ and $b$. 
\end{proof}

\begin{property}\label{parallelmakespan}
Suppose that there are $p$ independent chains \\ $C(m_1,i_1), C(m_2, i_2), C(m_3, i_3), \ldots C(m_p, i_p)$, where the length of the skinny jobs in any two different chains are different powers of 2.  Then the optimum makespan of these $p$ chains is $\Omega(2^{m_1} + 2^{m_2} + 2^{m_3} + \cdots + 2^{m_p})$.
\end{property}
We defer the proof to the Appendix~\ref{sec:appendix}.

\begin{lemma}\label{cleanschedule}
Let $\mathcal{C}$ be a set of chains. There could be precedence relations among the chains: if chain $C_1 \prec C_2$, then all jobs in chain $C_1$ must be finished before a job from chain $C_2$ starts. Given any algorithm $\mathcal{A}$ that processes chains $\mathcal{C}$ with makespan $T$, we can find an algorithm $\mathcal{A}'$ with makespan at most $2\cdot T$ with the following properties: 

(1) algorithm $\mathcal{A}'$ does not start any job while another job has already started but not finished, and 

(2) at any point in time, $\mathcal{A}'$ runs only one type of skinny job (i.e., all the skinny jobs that run at the same time have the same length).
\end{lemma}

\begin{proof}
From \hao{added lemma number} Property~\ref{immediatestartchain}, we can assume that algorithm $\mathcal{A^{'}}$ does not start any job while another job has already started but not finished. Let $t$ be a time-step when $\mathcal{A^{'}}$ starts a set of jobs, and let $t^{'}$ be the time-step by which all these jobs that start at time $t$ finish (Figure~\ref{fig:chains}(e)). Note that $\mathcal{A^{'}}$ does not start any job during the time interval $(t, t^{'}]$. Each of the jobs that starts at $t$ must be a skinny job (since a fat job cannot run with any other job) and the length of these jobs are powers of 2. Let $2^q$ be the maximum length of a job that starts at $t$ and finishes at $t^{'}$. 

We can then find another schedule where only the skinny jobs of length $2^q$ run in parallel, followed by only the skinny jobs (if exists) of length $2^{q-1}$ run in parallel, followed by only the skinny jobs (if exists) of length $2^{q-2}$ run in parallel, and so on (Figure~\ref{fig:chains}(f). Let this staged-execution finish at time $t^{''}$. Since all the jobs' lengths are power of $2$, we have $(t^{''} - t) \le 2\cdot (t^{'}-t)$. 

Applying the above procedure to all the start times $t$, the makespan is increased by at most a factor of 2, and the lemma is thus proved.

\end{proof}

\section{Hardness of $o( ( \log t_{max} )^\alpha )$-Factor Approximation for the Offline Problem}
\label{sec:lts}
In this section, we prove the hardness of approximation for the precedence-constrained resource scheduling problem, even when there is only one type of resource.

We do a reduction from the loading time scheduling problem (LTS) introduced by Bhatia et al.~\cite{LTSnotAPX}. 
First, we restate the LTS problem.  

\subsection{Revisiting LTS Problem}
\label{sec:lts-revisit}

Let $D = (V,E)$ be a DAG, where each node $v\in  V$ represents a job, and each edge $(u,v)$ represents that job $u$ must be completed before starting job $v$ (i.e., $u\prec v$). Let $|V| = n$, meaning there are $n$ jobs. Let $M$ be the set of $\rho$ machines $m_1, m_2, \ldots, m_\rho$. Each job $v$ can be performed only on a single machine $\delta(v)\in M$.  

Each machine $m_i$ has a loading time $\ell(m_i)$. The processing time of each job is 0. Performing a set of tasks consecutively on any given machine only requires paying the loading time once.  

A LTS instance is thus defined as $\Pi = (D, M, \delta, \ell)$.  

\paragraph{\textbf{ The objective.}} The objective is to partition the set of jobs $V$ into $k$ subsets $V_1, V_2, \ldots, V_k$, such that the following conditions hold:

\begin{enumerate}
\item For all $V_i$'s where $i \in [1,k]$, all the jobs $v$ that are put in $V_i$, must have the same machine assignment $\delta(v)$.
\item For each edge $(u,v)\in E$, if $u\in V_i$ and $v\in V_j$, then $i \le j$. 
\end{enumerate}

Let $m_i$ be the machine that all the jobs in $V_i$ can be executed. Then the cost of partition $V_1, V_2, \ldots, V_k$ is $\sum_{i =1}^k \ell(m_i)$.

The goal is to find the partition with minimum cost.  \myworries{start cutting here and cut up to Theorem 3.11}

\paragraph{\textbf{ An example.}} Let $x$, $y$, and $z$ be three jobs. Let job $x$ has a precedence over job $z$. Let $\delta(x) = \delta(z) = m_1$ and $\delta(y) = m_2$. A feasible LTS partition can put $x$ and $z$ in the same subset and the partition in this case could be, 
$V_1 = \{x, z\}$ and $V_2 = \{y\}$,  
with the total loading time given by $\ell(m_1) + \ell(m_2)$. 

On the other hand, if $x\prec y\prec z$, then we cannot put both $x$ and $z$ in the same subset $V_1$. 
 
In this case, a feasible partition would be $V_1 = \{x\}$, $V_2 = \{y\}$, and $V_3 = \{z\}$, with the total cost now equal to $\ell(m_1) + \ell(m_2) + \ell(m_1)$.  



\begin{rem}\hao{fixed citation}
In~\cite{LTSnotAPX}, the LTS problem is defined where a job can be performed only on a subset of machines. However, the authors proved in their paper that the hardness of approximation works, even if each job can be performed only on a single machine. Hence, we can assume that each job can be performed only on a single machine for the LTS instances that we consider for our hardness of approximation reduction.   
\end{rem}

The following hardness of approximation result of LTS is known from~\cite{LTSnotAPX}.


\LTSorig

In this section, we prove the following theorem.
\ResourceImplyLTS

\subsection{Modifying DAGs for LTS Problem}
The LTS problem cannot readily be used as its current definition for the hardness reduction---Consider the situation where there are three jobs $j_1, j_2$, and $j_3$. Let all three jobs can be assigned on the same machine $m$ and $j_1\prec j_2 \prec j_3$. Then an LTS solution, where all three jobs are put in the same partition, costs only the loading time of machine $m$ once (processing times of jobs are not taken into cost calculation). However, in the scheduling solution, if we run multiple jobs in a single machine, we need to run them sequentially, and thus their processing times add up.

To resolve this issue, we transform the input graph of LTS into another graph  called a \defn{conflict-free} graph with lesser precedence relations, which bypasses this issue: In a conflict-free graph, for any two jobs $j$ and $j^{'}$ with the same machine assignment $m$, if $j \prec j^{'}$, then there must be another job $j^{''}$ with a different machine assignment $m^{''}$ and $j \prec j^{''} \prec j^{'}$. In this case, any feasible LTS solution must put $j, j^{'}$, and $j^{''}$  in three different partitions. This allows us to relate the sum of processing times of $j, j^{'}$,  and $j^{''}$ in the resource scheduling problem to the sum of loading costs of the three partitions corresponding to $j, j^{'}$,  and $j^{''}$ in the LTS solution.
We then prove that the inapproximability results of LTS hold in the reduced intermediate LTS problem with conflict-free graphs. 

We still need to perform an additional transformation of the intermediate LTS problem with a conflict-free graph.
Specifically, the loading times of the machines may not be bounded. To ensure our reduction has polynomial in the input size, we want the loading times to be upper-bounded by a polynomial in the problem size, i.e., polynomial in the number of jobs and machines. 
We achieve this by reducing the intermediate problem to a conflict-free \defn{bounded} LTS problem and demonstrate that the same asymptotic inapproximability guarantee holds for this transformed problem. 

We then use our tool, \defn{chains} to reduce the conflict-free bounded LTS problem to our resource scheduling problem and prove the reduction.

There are two stages of transformations: (1) Creating a Conflict-Free DAG and (2) Creating a Conflict-Free LTS Instance with Bounded Loading Time.

\subsubsection{Stage-1 Transformation: Creating a conflict-free DAG}
An edge $(u,v)$ is called \defn{redundant} if there is a directed path of length at least 2 from $u$ to $v$.\haonew{moved sentence so it comes first}
Given a DAG $D = (V,E)$, we delete all the redundant dependency edges (removing redundant edges does not change the precedence relation among the nodes in the DAG).  Without loss of generality, we thus assume that the input DAG has no redundant edges. 
Call an edge $(u,v)$ \defn{bonded}, if $\delta(u) = \delta(v)$, i.e., both tasks $u$ and $v$ can only be assigned on the same machine. Call a DAG \defn{conflict-free}, if there are no bonded edges in the DAG. \hao{add a comment that we henceforth only talk about LTS instances with $|M(v)|=1$}  

Given a DAG $D$, we now present an algorithm to construct a conflict-free DAG $D^{'}$ such that the optimal loading time of tasks in $D$ is the same as that of $D^{'}$.

\paragraph{\textbf{ Algorithm for constructing a conflict-free DAG}}

Our algorithm proceeds in iterations, where in each iteration $i$, we \defn{resolve} a bonded edge $(u,v)$ if there is any. 
The algorithm terminates when no bonded edges remain, meaning the DAG has become conflict-free.

Now, let’s explain how we resolve a bonded edge in an iteration.
Let $D_{i-1} = (V, E_{i-1})$ be the input DAG for iteration $i$. Remove any redundant edges if present in $D_{i-1}$. Let $D_{i} = (V, E_i)$ be the modified DAG after resolving bonded edge $(u,v)$ in iteration $i$. The vertex set of $D_{i-1}$ and $D_i$ are the same. Only the edges $E_{i}$ updated in the following way (see Figure~\ref{fig:lts-conflict}):   
\begin{enumerate}
\item For each incoming edge $(w,u)$ to $u$, create a new edge $(w,v)$ in $D_{i}$.
\item Delete the edge $(u,v)$ in $D_{i}$.
\item Keep all the edges of $D_{i-1}$, except for $(u,v)$ in the new DAG $D_{i}$. 
\end{enumerate}

\begin{figure}
  \begin{center}
    \includegraphics[width=0.45\textwidth]{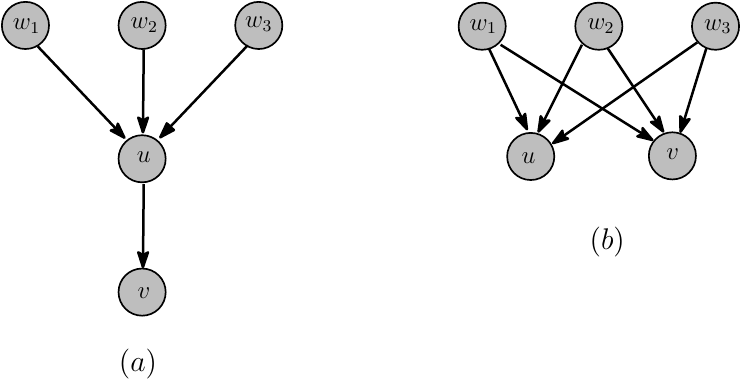}
  \end{center}
  \caption{Part (a): edge $(u,v)$ is a bonded edge. Part (b) Resolving bonded edge $(u,v)$.}
  \label{fig:lts-conflict}
\end{figure}

\begin{lemma}
\label{lem:conflict-free-time}
Let the input DAG $D = D_0$ have $n$ vertices. Let after $t$ iterations, $D_t$ becomes a conflict-free DAG (it has no bonded edges). Then $\textsc{opt-lts}(D_{0}) = \textsc{opt-lts}(D_t)$ and $t \le n^2$.
\end{lemma}

We prove Lemma~\ref{lem:conflict-free-time} by proving the following lemmas.

\begin{lemma}
Let $D_{i-1}$ be the input DAG for $i$-th iteration and $D_i$ be the output DAG after $i$-th iteration. Then $\textsc{opt-lts}(D_{i-1}) = \textsc{opt-lts}(D_i)$. 
\end{lemma}
\begin{proof}
Let bonded edge $(u,v)$ be resolved in the $i$-th iteration. Let $V_1, V_2, \ldots, V_k$ be the partition of tasks in $D_{i-1}$ by a feasible LTS solution. We can argue  that $V_1, V_2, \ldots, V_k$ is a feasible partition of the tasks in $D_i$. 
Since no new precedence constraints are created in $D_{i}$, any solution that is feasible for $D_{i-1}$, is also feasible for $D_i$.

Similarly, let $V^{'}_1, V^{'}_2, \ldots, V^{'}_{k^{'}}$ be the partition $\Pi$ of tasks in $D_{i}$ by a feasible LTS solution with loading time $X$. 
We show that there exists a feasible partition of the tasks in $D_{i-1}$ with loading time at most $X$. 
Note that $D_{i-1}$ has one extra precedence constraint---if task $u\in V^{'}_{j}$ and task $v\in V^{'}_{j^{'}}$, then $j^{'} \ge j$. 
If the LTS partition $\Pi$ for $D_i$ puts task $u$ in $V^{'}_j$ and task $v$ in $V^{'}_{j^{'}}$ where $j^{'} \geq j$, then this is a feasible solution for $D_{i-1}$.
On the other hand, if the LTS partition $\Pi$ for $D_i$ puts task $u$ in $V^{'}_j$ and task $v$ in $V^{'}_{j^{'}}$ where $j^{'} < j$, we can create another LTS partition $\Pi^{'}$ with no worse loading time---we remove task $v$ from $V^{'}_{j^{'}}$ and add task $v$ to $V^{'}_{j}$. All the precedence constraints are still met in $\Pi^{'}$ and the loading time of $\Pi^{'}$ does not increase from the previous solution $\Pi$.
Since the new solution $\Pi^{'}$ puts both $u$ and $v$ in the same subset, this is also a feasible solution for $D_{i-1}$.  

The lemma is thus proved.
\end{proof}

We now show that if there is no precedence relation (a directed path) between node $x$ and $y$ in an iteration $j$, then there is no precedence relation between $x$ and $y$ in any iterations $j^{'} > j$. Using this, we show that once a bonded edge is deleted in an iteration, it is never inserted back in future iterations.

\begin{lemma}
\label{lem:no-path-added}
If there is no directed path from a node $x$ to $y$ in an iteration $j$, then there is no directed path between $x$ and $y$ in any iterations $j^{'} > j$.
\end{lemma}
\begin{proof}
    We show that  if there is no directed path from a node $x$ to $y$ in an iteration $j$, then there is no directed path between $x$ and $y$ in iteration $j+1$ and the result will follow by induction.
    Suppose there is a directed path $\langle x=x_0,x_1,.., x_m =y\rangle$ in iteration $j+1$.   
    Let the bonded edge $(u,v)$ be resolved in the $(j+1)$-th iteration.
    Define $f(a,b) = \langle a, b \rangle$ if there is an edge $(a,b)$, otherwise, define $f(a,b)$ to be an arbitrary path from node $a$ to node $b$, if such a path exists. If the edge $(x_i, x_{i+1})$ is present in iteration $j$, then $f(x_i, x_{i+1}) = \langle x_i, x_{i+1}\rangle$. Otherwise, the edge  $(x_i, x_{i+1})$ was added in iteration $j+1$. Then $x_i$ must be an in-neighbour of $u$, and $x_{i+1}=v$ in iteration $j$. Thus, $f(x_i, x_{i+1}) $ must be the path $\langle x_i,u, x_{i+1} \rangle$ in iteration $j$.
    In either case, the path $f(x_i, x_{i+1})$ is present in iteration $j$. Thus the $x$-$y$ path  $\langle f(x_0,x_1), f(x_1,x_2),.., f(x_{m-1}, x_m) \rangle$ is present in iteration $j$.
\end{proof}

\begin{lemma}
\label{lem:bonded-removal-once}
If a bonded edge is deleted in an iteration, then the edge is never inserted in any future iterations.  
\end{lemma}
\begin{proof}
Let  $(u,v)$ be a bonded edge in DAG $D_{i-1}$ that is resolved in the $i$-th iteration. Since DAG $D_{i-1}$ does not contain any redundant edges, there is no directed path from $u$ to $v$. Also, after iteration $i$, the edge $(u,v)$ is deleted. Hence, after iteration $i$, there is no precedence relation between $u$ and $v$.  

In Lemma~\ref{lem:no-path-added}, we showed that if there is no precedence relation (a directed path) from a node $x$ to $y$ in an iteration $j$, then there is no precedence relation between $x$ and $y$ in any iterations $j^{'} > j$.

Hence, there can not be a new edge $(u,v)$, since $u$ and $v$ remain independent (no precedence) in all the iterations following $i$. The lemma is thus proved.
\end{proof}

We upper bound the number of iterations in the following lemma.

\begin{lemma}
Let the input DAG $D = D_0$ have $n$ vertices. Let after $t$ iterations, $D_t$ has no bonded edges. Then $t \le n^2$.
\end{lemma}
\begin{proof}
In Lemma~\ref{lem:bonded-removal-once}, we show that once a bonded edge is removed, the edge never gets inserted back in future. Since, there are at most $ \binom{ n}{2}$ many possible edges, the algorithm finishes within $n^2$ iterations.
\end{proof}

\paragraph{\textbf{ An important property of conflict-free DAGs.}} 
We now present an important property of conflict-free DAGs in the following lemma.
\begin{lemma}
\label{lem:conflict-prop}
Let $\Pi = (D, M, \delta, \ell)$ be a LTS instance where $D$ is a conflict-free DAG. Let $S$ be an arbitrary solution to $\Pi$ where $P = \{V_1, V_2, V_3, \ldots V_{k} \}$ is the partition of the jobs in solution $S$. Then in every $V_i$, for any two jobs $u, v \in V_i$, there is no precedence relation between $u$ and $v$. 
\end{lemma}

\begin{proof}
Since $u,v \in V_i$, $\delta( v ) = \delta(u) $. 
Since $G$ is conflict free, this implies that there is no precedence relation between $u$ and $v$.
\end{proof}

\subsubsection{Stage-2 Transformation: Creating a conflict-free LTS instance with bounded loading time}

Let $\Pi = (D, M, \delta, \ell)$ be a conflict-free LTS instance. We perform the following transformation on $\Pi$ to get the loading time of each machine upper bounded by a polynomial of the number of jobs and machines.

\paragraph{\textbf{ Step 1: Creating a bounded loading time instance $\Pi^{'}$.}}
Given a conflict-free LTS instance $\Pi = (D, M, \delta, \ell)$, we create a new LTS instance $\Pi^{'} = (D^{'}, M^{'}, \delta^{'}, \ell^{'})$ as follows.

Let $M = \{m_1,m_2, \ldots ,m_\rho\}$ be the set of machines with non-decreasing loading time (i.e., $\ell(m_i) \le \ell(m_{i+1})$ for $i \in [1, \rho -1]$). Let $\hat{i} \in [1,\rho-1]$ be the largest index such that $\ell(m_{\hat{i}+1}) > n\cdot \ell(m_{\hat{i}})$, where $n$ is the number of vertices in $D$. If no such index $\hat{i}$ exists, then $\Pi^{'} = \Pi$. Otherwise, create a new LTS instance $\Pi^{'}$ as follows.

\begin{itemize}
\item The set of machines $M^{'}\subseteq M$ are those machine whose loading time is at least as large as $m_{\hat{i}}$, i.e., $M^{'} =  \{m \in M | \ell(m) \ge \ell(m_{\hat{i}})\}$.
\item Vertex set $V^{'}\subseteq V$ is the set of nodes $v$ such that $\delta(v)\in M^{'}$. 
\item The precedence relations $E^{'}$ among the nodes in $V^{'}$ are inherited from DAG $D$: for any two nodes $u, v\in V^{'}$, if $u\prec v$ in $D$ (i.e., there is a directed path from node $u$ to $v$ in $D$), then directed edge $(u,v)\in E^{'}$.
\item The machine assignment function $\delta^{'}: V^{'} \rightarrow M^{'}$ in instance $\Pi^{'}$ is defined as follows: $\delta^{'}(v) = \delta(v)$ for each $v\in V^{'}$.
\item The machine loading time function $\ell{'}: M^{'} \rightarrow \mathbb{N}$ in instance $\Pi^{'}$ is defined as follows: $\ell^{'}(m) = \ell(m)$ for each $m\in M^{'}$.
\end{itemize}

Given a solution $S^{'}$ for instance $\Pi{'}$, we can construct a solution $S$ for instance $\Pi$ as follows. Let $V^{'}_1, V^{'}_2, V^{'}_3, \ldots, V^{'}_{k^{'}}$ be the partition of jobs in $V^{'}$. For each job $v \in V \setminus V^{'}$, create a new subset $\{v\}$. The new partition of solution $S$ is the collection of the original subsets in the partition from $S^{'}$ and the newly created subsets.  Importantly, we can preserve the LTS ordering of subsets in this new partition. 
\myworries{Do we need to be more detailed? each $v$ can go between different $V_i, V_j$ and there is a proper way to order those that go between ?  }

\begin{lemma}
Let $S^{'}$ be a solution  to the reduced conflict-free LTS instance $\Pi^{'}$ with bounded loading time where $\textsc{cost}_{\Pi^{'}}(S^{'})$ be its cost. Let $S$ be the corresponding solution to the original LTS instance $\Pi$ where $\textsc{cost}_{\Pi}(S)$ be its cost. Then, $\textsc{cost}_{\Pi}(S) \le 2\cdot \textsc{cost}_{\Pi^{'}}(S^{'})$.
\end{lemma}

\begin{proof}
Let $\hat{i} \in [1,\rho-1]$ be the largest index such that $\ell(m_{\hat{i}+1}) > n\cdot \ell(m_{\hat{i}})$, where $n$ is the number of vertices in $D$. Then, $m_{\hat{i}}$  is the smallest loading time machine in $\Pi^{'}$. Thus, $\textsc{cost}_{\Pi^{'}}(S^{'}) \ge \ell(m_{\hat{i}+1})$. 
From the construction of $V^{'}$, any job $v$ that is in $V$, but not in $V^{'}$, must be assigned on a machine whose loading time $\ell(\delta(v)) \le \ell(m_{\hat{i}})$. 
Clearly, the number of such jobs that are in $V$, but not in $V^{'}$ is at most $n$. Hence, $\textsc{cost}_{\Pi}(S) \le \textsc{cost}_{\Pi^{'}}(S^{'}) + n\cdot \ell(m_{\hat{i}}) \le \textsc{cost}_{\Pi^{'}}(S^{'}) + \ell(m_{\hat{i}+1}) \le 2 \cdot \textsc{cost}_{\Pi^{'}}(S^{'})$. The lemma is thus proved.   
\end{proof}

\paragraph{\textbf{ Step 2: Creating a LTS instance $\Pi^{''}$ with powers of 2 loading times.}}
Given a LTS instance $\Pi^{'} = (D^{'}, M^{'}, \delta^{'}, \ell^{'})$, we create a new LTS instance $\Pi^{''} = (D^{'}, M^{'}, \delta^{'}, \ell^{''})$ where for each machine $m$, we round $\ell^{'}(m)$ to the nearest power of 2 to get $\ell^{''}(m)$.

From the construction of $\Pi^{''}$, the following lemma is immediate.
\begin{lemma}\label{removesmalljob}
Let $S^{''}$ be a solution to the reduced LTS instance $\Pi^{''}$ where $\textsc{cost}_{\Pi^{''}}(S^{''})$ be its cost. Let $S$ be the corresponding solution to the LTS instance $\Pi$ where $\textsc{cost}_{\Pi^{'}}(S^{'})$ be its cost. Then, $\textsc{cost}_{\Pi^{''}}(S^{''}) \le 2\cdot \textsc{cost}_{\Pi^{'}}(S^{'})$.
\end{lemma}

We now state the theorem for hardness of approximation on the reduced LTS problem instances $\Pi^{''}$. 

 \begin{theorem}\label{LTSmodified}
 For some constant $\alpha >0$, 
 no polynomial-time algorithm can achieve a $\rho^ \alpha\cdot (1/4)$-approximation ratio  for the conflict-free bounded LTS problem in the restricted case when $|M(v)|=1$ and all load times are powers of 2 unless P= NP. 
 \end{theorem}

\begin{proof}  
We prove this theorem by contradiction. Given an arbitrary LTS instance $\Pi$, we reduce $\Pi$ to $\Pi^{''}$ by applying step 1 and step 2 as described above. Suppose that there is a polynomial-time $\rho^ \alpha\cdot (1/4)$-approximation algorithm $\mathcal{A}$ on the reduced LTS problem instances $\Pi^{''}$. Then algorithm $\mathcal{A}$ gives a solution $S^{''}$ to $\Pi''$ satisfying $\textsc{cost}_{\Pi''}(S)  \le \rho^ \alpha\cdot (1/4) OPT_{\Pi''} $. We thus get the following.
$$\textsc{cost}_{\Pi'}(S^{'}) \leq \textsc{cost}_{\Pi''}(S^{''}) < \rho^ \alpha\cdot (1/4) \textsc{opt}_{\Pi''}  \leq  \rho^ \alpha\cdot (1/2) \textsc{opt}_{\Pi'}.$$ 
(
We lose 1/2 factor going from $\Pi^{''}$ to $\Pi^{'}$ due to rounding).

From the discussions above, given a solution $S^{'}$ to $\Pi^{'}$, we can find (in polynomial time) a solution $S$ to $\Pi$ of cost at most $2 \cdot \textsc{cost}_{\Pi^{'}}(S).$ (From Lemma~\ref{removesmalljob}, we lose 1/2 factor going from $\Pi^{'}$ to $\Pi$ due to adding back small jobs).  Then, 
$$\textsc{cost}_{\Pi}(S) \leq 2  \textsc{cost}_{\Pi'}(S^{'}) \le \rho^ \alpha \textsc{opt}_{\Pi'} \leq \rho^ \alpha \textsc{opt}_{\Pi}.$$ This contradicts \autoref{LTSoriginal}, which says no polynomial-time algorithm $\mathcal{A}$ can achieve $\rho^{\alpha}$-approximation ratio on every LTS instance $\Pi$, unless P=NP.
The theorem is thus proved.
\end{proof}

\paragraph{\textbf{ Normalizing the loading time.}} 
Given any reduced LTS instance $\Pi^{''}$, we scale the loading times so that the minimum loading time is 1. 

\paragraph{\textbf{ An important property: Bounded loading time ensures the reduction size is polynomial.}} Let $m_1, m_2, m_3, \ldots, m_{\rho}$ be the sequence of machines with non-decreasing loading times in a normalized bounded LTS instance $\Pi^{''}$. From the construction, we have $\ell(m_{i+1}) \le n\cdot \ell(m_i)$ for all $i \in [1, \rho-1]$. Since, after normalization, $\ell(m_i) = 1$, the maximum loading time is $n^{\rho}$. In our reduction, $\rho$ is a constant (recall that the hardness result from~\cite{LTSnotAPX} works even if the number of machines $\rho$ is a fixed constant). Hence, the loading times are upper bounded by a polynomial in $n$. 

\subsection{Approximation-Preserving Reduction from LTS to Resource Scheduling}
We will now show an approximation preserving reduction from conflict-free bounded LTS to resource scheduling.

Given a conflict-free bounded LTS instance $ \Pi = (D,M,\delta,\ell) $, we construct a resource scheduling instance $\textsc{rs}_{\Pi}$ as follows. Let $m_1, m_2, m_3, \ldots m_\rho$ be the sequence of machines with non-decreasing leading time, where $\ell(m_{i+1}) \ge \ell(m_{i})$ for $i\in[1, \rho -1]$.
For $v \in D$, let $i(v)$ denote the index of the machine that $v$ can be scheduled on, that is, $\delta(v)=m_{i(v)}$.
For each node $v \in D$, construct a chain $C_v = C(f(i(v)), i(v))$ where $f(i(v)) = \rho + \log \ell(m_{i(v)})$. Recall that, a chain $C(m,i)$ is a concatenation of $2^m/2^i$ tuples where each tuple is a skinny job of length $2^i$, followed by a fat job. 
Hence, chain $C_v$ is a concatenation of $2^{\rho + \log \ell(m_{i(v)}) }/ 2^{i(v)} = 2^{\rho -i(v) } \cdot \ell(m_{i(v)})$ tuples where each tuple is a skinny job of length $2^{{i(v)}}$, followed by a fat job. 
For each $(u,v) \in D$, draw a directed edge from the sink node of $C_u$ (which is a fat job) to the source node of $C_v$ (a skinny job).
This forms our resource scheduling instance $\textsc{rs}_{\Pi}$.  

Note that if an LTS instance $ \Pi$ has $\rho$ machines, then the resulting resource scheduling instance $\textsc{rs}_{\Pi}$ has maximum job size at most $2^\rho$. However, in our reduction, $\rho$ is a fixed constant. 

\begin{rem}
\label{rem:lts-skinny}
In the resource scheduling instance $\textsc{rs}_{\Pi}$, each skinny job of length $2^i$ is part of a chain $C(f(i), i)$ where $f(i) = \rho + \log \ell(m_i)$.
\end{rem}

\paragraph{\textbf{ Size of the reduced resource scheduling instance $\textsc{RS}_{\Pi}$.}} Let $n$ be the number of jobs and $\rho$ be the number of machines in the conflict-free bounded LTS instance $\Pi$. In this problem, $\rho$ is a fixed constant. 

For each node $v$ in $\Pi$, we create a chain $C(f(i), i)$, where job $v$ can be run only on machine $m_i$. The maximum number of nodes in a chain is $2^{f(i)}$, where $f(i) = \rho + \log \ell(m_i)$. The loading time of any machine is upper bounded by $n^{\rho}$ due to the bounded loading time construction. Thus, a chain can contain at most $2^\rho \cdot n^\rho$ nodes. 

There are $n$ jobs in the LTS instance, so the total number of nodes in the reduced resource scheduling instance is $n \cdot 2^\rho \cdot n^\rho = 2^\rho \cdot n^{\rho + 1}$.  

Recall that $\rho$ is a fixed constant, and the hardness of approximation results for LTS still holds even if the number of machines $\rho$ is a constant (see Theorem~\ref{LTSoriginal}). Thus, our reduced resource scheduling instance \textbf{size is still a polynomial in $n$}, since $\rho$ is a constant. 

We first prove the following lemma to prove the main theorem of this section.
\begin{lemma}\label{lem:ResourceSchedulingeqLTS}
Let $\Pi = (D,M,\delta, \ell)$ be a LTS instance where the following conditions are satisfied: $D$ is a conflict-free DAG, the maximum loading time of a machine is $n^\rho$ where $n$ is the number of jobs in $D$ and $\rho$ is the number of machines, all loading times are powers of 2, and the number of machines $\rho$ is a constant. Let $\textsc{rs}_{\Pi}$ be the reduced resource scheduling instance. Then the following statements hold.

\textbf{{Part 1:}} Given  a solution $V_1,V_2,.., V_k$ to the LTS  instance $\Pi$ of cost $r$, we can find in polynomial time a solution to the reduced resource scheduling instance $\textsc{rs}_{\Pi}$ that achieves makespan $2^\rho \cdot r$.

\textbf{Part 2: } Given  a solution to $\textsc{rs}_{\Pi}$ achieving  makespan $ \tau$, we can find in polynomial time  a solution to the LTS  instance $\Pi$ of cost $ \tau /2^{\rho -1}$.
\end{lemma}


\paragraph{\textbf{ Proof of Part 1 of Lemma~\ref{lem:ResourceSchedulingeqLTS} (LTS solution $\rightarrow$  scheduling solution).}} 
\begin{proof}
Let $V_1,V_2,.., V_k$ be the partition of jobs in the LTS solution. From the LTS problem definition, all the jobs in each $V_i$ are scheduled to a single machine $m_i$. The cost of the LTS solution is thus $r = \ell(m_1) + \ell(m_2) +\cdots +\ell(m_{k})$. Since, the input DAG $D$ is conflict-free, in every $V_i$, for any two jobs $u, v\in V_i$, there is no precedence relation between $u$ and $v$. Consider all the chains in the resource scheduling instance $\textsc{rs}_{\Pi}$ that correspond to the jobs in $V_i$. Since all these jobs in $V_i$ must schedule in the same machine $m_i$, each chain is of type $C(\rho + \log \ell(m_i), i)$. From \hao{changed} Property~\ref{skinnyjobpar}, a set of independent chains (no precedence relation between two jobs from different chains) with the same length skinny jobs can be scheduled in parallel and can be finished in time that is the length of the longest chain. Hence, all the chains corresponding to jobs in $V_i$ can be finished in a duration of length $2^\rho \cdot \ell(m_i)$.


We first run the jobs only from the chains corresponding to $V_1$, after the chains corresponding to $V_1$ are finished, run the jobs only from the chains corresponding to $V_2$, after the chains corresponding to $V_2$ are finished, run the jobs only from the chains corresponding to $V_3$, and so on.

The makespan of this schedule is $2^\rho \cdot \ell(m_1) + 2^\rho \cdot \ell(m_2) + \cdots + 2^\rho \cdot \ell(m_{k})= 2^{\rho} \cdot r$.
\end{proof}

\paragraph{\textbf{ Proof of part 2 of  Lemma~\ref{lem:ResourceSchedulingeqLTS} (scheduling solution $\rightarrow$  LTS solution).}}
\begin{proof}
Let $\mathcal{A}$ be an algorithm for the resource scheduling problem that gives a schedule $S$ with makespan $\tau$ on the reduced resource scheduling instance $\textsc{rs}_{\Pi}$. Note that instance $\textsc{rs}_{\Pi}$ can be thought as a DAG where each node $u$ is a chain $C_u$, and if there is an edge from node $u$ to $v$, then chain $C_u \prec C_v$. From \hao{added lemma number} Lemma~\ref{cleanschedule}, we can find an algorithm $\mathcal{A}'$ that gives a schedule $S^{'}$ with makespan at most $2 \cdot \tau$ with the following two properties:

(1) algorithm $\mathcal{A}'$ does not start any job while another job has already started but not finished, and 

(2) at any point in time, $\mathcal{A}'$ runs only one type of skinny jobs (i.e., all the skinny jobs that run at the same time have the same length). 
\vspace{-3mm}
~\paragraph{\textbf{ Creating a new schedule $S^{''}$ where jobs of a chain run without any gap.} }
Let $\mathcal{S}^{'}$ be the schedule by algorithm $\mathcal{A}^{'}$. From Remark~\ref{rem:lts-skinny}, in the DAG of resource scheduling instance, each skinny job of length $2^i$ is part of a chain $C(f(i), i)$, where the length of the chain is $2^{f(i)}$.
Let $j(i)$ denote the number of intervals that run a job of length $2^i$. \hao{what term do we use intervals, epoch, period saying time is potentially confusing}
For $1 \leq j \leq  j(i) 2^i / 2^{f(i)}$, let $t_{i,j}$ be the timestep in schedule $\mathcal{S}^{'}$ when skinny jobs of length $2^i$ are scheduled for the $j\cdot (2^{f(i)}/2^i)$-th time (i.e., the sum of lengths of skinny-jobs each of length $2^i$ until time $t_{i,j} +2^i$ is $j\cdot 2^{f(i)}$). \myworries{ whats in brackets should be sum of lengths of skinny-jobs each of length $2^i$ until time $t_{i,j} +2^i $ is $j\cdot 2^{f(i)}$  ? $t_{i,j}$ is start of interval not end?   } Note that $t_{i,1} < t_{i,2} < t_{i,3}< \cdots$ for each skinny-jobs of length $2^i$. 
Besides, given two skinny jobs of lengths $2^{i}$ and $2^{i^{'}}$, \hao{You don't need to say given two jobs it's given any two integers $ 1\leq i,i' \leq \rho$} and two integers $j$ and $j^{'}$, there is an ordering of $t_{i,j}$ and $t_{i^{'}, j^{'}}$ (i.e., either $t_{i,j} < t_{i^{'}, j^{'}}$ or $t_{i^{'}, j^{'}} < t_{i,j}$) in schedule $\mathcal{S}^{'}$. 
Since there is a total ordering $t_{i_1,j_1} < t_{i_2,j_2} <\ldots $ of all the $t_{i,j}$'s, create a sequence $X$ of integers $i$ in the order of $t_{i,j}$'s (e.g., if $t_{5,1} < t_{3,1} < t_{5,2} < t_{5,3} < t_{4,1} < t_{3,2} < \cdots$, then the sequence $X$ of integers is $(5,3,5,5,4,3,\ldots)$, where $X[1] = 5, X[2] = 3, X[3] =5, X[4] = 5$ etc.). 


Given $X$, we create a new schedule $\mathcal{S}^{''}$ from $X$ as follows: Schedule only ready-to-be-scheduled chains of type $C(f(X[1]), X[1])$ in parallel (call this time interval as 1st epoch), followed by only ready-to-be-scheduled chains of type $C(f(X[2]), X[2])$ in parallel (call this time interval as 2nd epoch), followed by only ready-to-be-scheduled chains of type $C(f(X[3]), X[3])$ in parallel (call this time interval as 3rd epoch), and so on.  
From the construction, schedule $S^{''}$ obeys the precedence relation.

\vspace{-3mm}
\paragraph{\textbf{ Each chain $C$ is run in schedule $S^{''}$.}}
For each chain $C$ in the DAG of the reduced resource scheduling instance, we show that $C$ is scheduled in $S^{''}$. 
Let $C$ be any chain.
Let $C(f(i_C)),i_C)$ be the type of $C$.
Let in schedule $S^{'}$, chain $C$ starts at time $t_s(C)$ and finishes at time $t_f(C)$ (since, $S^{'}$ is a feasible schedule, it must schedule $C$). 
Since chain $C(f(i_C)),i_C)$ starts and finishes in $[t_s(C), t_f(C)], $, there are at least $ 2^{f(i)} /2^i $ intervals in $[t_s(C), t_f(C)], $ where a job of length $2^i$ is scheduled. 
Hence there must be a $t_{i_C,j_C} \in [t_s(C), t_f(C) -1]$.
Find the index $k_C$ in the integer-sequence $X$ that corresponds to $t_{i_C,j_C}$. Hence, $X[k_C] = i_C$.

We claim that for any chain $C$ all predecessors of $C$ are completed by the $k_C-1$-th epoch.
We prove by this induction on the number of ancestors of $C$.
If $C$ has no predecessor, then the statement is clear.
Let $C'$ be any predecessor of $C$.
Then by induction all predecessors of $C'$ are complete after the $ k_{C'}-1 $-th epoch. So $C'$ will be completed by the $ k_{C'} $-th epoch. 
Since $C'$ must finish before $C$ starts in schedule $S'$,
$ t_{i_{C'},j_{C'} } \leq t_f(C') -1 < t_s(C) \leq  t_{i_C,j_C} $. 
Thus $ k_{C'}  \leq k_C -1$. 
Since this holds for all predecessors, this means that all predecessors of $C$ are completed by the $k_C-1$-th epoch, thus completing the induction.

This means that each $C$ will be completed by the $k_C$-th epoch. Since $k_C$ exists for all $C$, this means the schedule finishes.

\paragraph{\textbf{ The length of $S^{''}$ is not increased.}}
The number of $ k $ such that $X[k]=i$ is $ \lfloor  j(i) 2^i / 2^{f(i)} \rfloor $. 
Each step scheduling $C(f(X[k]), X[k])$ with $X[k]=i$ takes time $ 2^{f(i)} $.
Thus, the length of the new schedule is $ \sum_{i=1}^\rho \sum_{k: \ \ X[k]=i } 2^{f(i)} \leq  \sum_{i=1}^\rho  j(i) 2^i / 2^{f(i)}  2^{f(i)} = \sum_{i=1}^\rho  j(i) 2^i  $.

For each $i=1,2,.., \rho$, the original schedule has $j(i)$ intervals of length $2^i$ when it processes jobs of size $2^i$.
The length of the schedule $\mathcal{S}^{'}$  is $\sum_{i=1}^\rho  j(i) 2^i $.
Thus, the length of $S^{''}$ is at most that of $\mathcal{S}^{'}$.
\end{proof}

We now prove the main theorem of this section.

\paragraph{\textbf{ Proof of Theorem~\ref{thm:lts-new}.}}
\begin{proof}
Assume that there is a $ \frac{1}{8}  (  \log t_{\max} )^\alpha   $-approximation algorithm \textsc{alg} for resource scheduling problem. 
We show that  \autoref{lem:ResourceSchedulingeqLTS}  gives a reduction from  finding a $ \frac{1}{4} \rho^\alpha $-approximation for the LTS problem to finding a $ \frac{1}{8}  (  \log t_{\max} )^\alpha  $-approximation resource scheduling as follows.

Given a LTS instance $ \Pi = (D, M, \delta, \ell) $ with optimal cost $r$, the first part of \autoref{lem:ResourceSchedulingeqLTS} implies that there exists a solution to resource scheduling instance $\textsc{rs}_{\Pi}$ achieving makespan $r \cdot 2^\rho$. Hence, the optimal makespan on $\textsc{rs}_{\Pi}$ is at most $r\cdot 2^\rho $. 
Recall that the maximum job size $t_{\max}$ of $\textsc{rs}_{\Pi}$ is at most $2^\rho$.
If there exists a polynomial-time $  \frac{1}{8}  (  \log t_{\max} )^\alpha $-approximation algorithm \textsc{alg} for the resource scheduling problem, then \textsc{alg} must find a solution to $\textsc{rs}_{\Pi}$ of makespan at most $  \frac{1}{8}  \rho^\alpha \cdot r \cdot 2^{\rho}$ since the optimal makespan on $\textsc{rs}_{\Pi}$ is at most $r \cdot 2^\rho$. 

The second part of \autoref{lem:ResourceSchedulingeqLTS} indicates that, given a solution to $\textsc{rs}_{\Pi}$ with makespan $  \frac{1}{8}  \rho^\alpha \cdot r \cdot 2^{\rho}$, we can obtain a solution to LTS instance $\Pi$ in polynomial-time with cost $( \frac{1}{8}  \rho^\alpha \cdot r \cdot 2^{\rho})/(2^{\rho-1}) = \frac{1}{4}  \rho^\alpha \cdot r$. 
Since, we started this argument with the assumption that the optimal cost on LTS instance $\Pi$ is $r$, this implies that we can find a $\frac{1}{4}  \rho^\alpha$-approximate solution to $ \Pi $ in polynomial time using a $ \frac{1}{8}  (  \log t_{\max} )^\alpha $-approximation algorithm \textsc{alg} of resource scheduling problem (by first, reducing the SCS instance $\Pi$ to a resource scheduling instance $\textsc{rs}_{\Pi}$, then, finding a $\frac{1}{8}  (  \log t_{\max} )^\alpha $-approximate solution to $\textsc{rs}_{\Pi}$, and then converting the solution of $\textsc{rs}_{\Pi}$ to a solution to $\Pi$---all three steps are done in polynomial time).  

From Theorem~\ref{LTSmodified}, we know that there is no polynomial-time algorithm with approximation ratio better than $\rho^\alpha$ unless P=NP. 
However, we just argued that if there is a polynomial-time $\frac{1}{8}  (  \log t_{\max} )^\alpha$-approximation algorithm for the resource scheduling problem, then there is a polynomial-time $\frac{1}{4}  (  \log t_{\max} )^\alpha < \frac{1}{4} \rho^\alpha$-approximation algorithm for bounded LTS problem, which contradicts Theorem~\ref{LTSmodified}. 
The theorem is thus proved.
\end{proof}

This implies the following corollary.

\begin{corollary}
   There is no $O(1)$-factor approximation for the precedence-constrained resource scheduling problem unless P = NP. 
\end{corollary}

\subsection{Lower Bound for Resource Scheduling Based on the Number of Jobs}
In this section, we prove the following hardness of approximation result for the precedence-constrained resource scheduling problem, based on the number of jobs in the instance.
\ResourceImplyLTSJobs
To do so we look at the details of how \cite{LTSnotAPX} shows their hardness of approximation result for LTS.
\begin{theorem}\label{SCSnotapxx} \cite{binSCSapx}  
    There exists a constant $\bar{c} >1 $ such that if there exists a polynomial-time $\bar{c}$-approximation algorithm for binary SCS, then  $P=NP$.
\end{theorem}

Bhatia et. al. \cite{LTSnotAPX} define a \emph{restricted SCS} instance to be an SCS instance where there are no consecutive runs of the same letter in any sequence.
Let $U$ be an instance of the SCS problem over the binary alphabet. 
One can construct an instance of restricted SCS from $U$ with double the alphabet size by replacing each letter $a$ by the sequence $a , a' $, where each $a'$ is a letter not used in $U$ and $a' \neq b'$ for different letters $a$ and $b$. 
Their construction shows that for some constant $c>1$, restricted SCS   over an alphabet of size 4 cannot be approximated within factor $c$.

\begin{theorem}\label{SCSnotapxx2} \cite{LTSnotAPX}
    There exists a constant $c>1$ such that if there exists a polynomial-time $c$-approximation algorithm for restricted SCS over an alphabet of size 4, then  $ P=NP $.
\end{theorem}


To show a lower bound for approximating resource scheduling, we need to analyze the details of how  \cite{LTSnotAPX} shows that LTS is hard to approximate. The following definitions are from \cite{LTSnotAPX}.

\begin{definition}\cite{LTSnotAPX}
     An LDAG is an acyclic digraph for which each vertex is
labeled by a single letter from a given alphabet.
\end{definition} 
\begin{definition}\cite{LTSnotAPX}
    A minimal supersequence $z$ of an LDAG is defined as
follows.
If the LDAG is empty, then $z= \emptyset$.
Let $a$ be the first letter of $z$, i.e. $z = a  \cdot z_1$, then some indegree 0
node in the LDAG is labeled with $a$, and $z_1$ is a minimal supersequence of
the LDAG obtained by deleting all indegree 0 nodes that have label $a$.
\end{definition} 
\begin{definition}\cite{LTSnotAPX}
    A supersequence of an LDAG is any sequence that
contains a minimal supersequence of the LDAG as a subsequence.
\end{definition} 
For an  LDAG $Y$, let $|Y|$ refer to the number of vertices in $Y$. 
For a set of sequences $X$, let $|X|$ refer to the total number of letters in all sequences in $X$.

We will refer to the LDAG instance $Y$ as the problem of finding the minimum length supersequence for $Y$.

\paragraph{\textbf{Defining LTS problem on LDAGs.}} We define the LTS problem on an LDAG as follows. 
There is a machine $m(a)$ for each letter $a$ of the alphabet, each machine has loading time 1.
Each node $v$ of the LDAG can be served only on the machine $m( q_v )$ where $v$ is labeled by the letter $q_v$.  
For a conflict-free LDAG, finding a minimum length supersequence is equivalent to solving the LTS problem on the LDAG, that is, for each partition $V_1,V_2,\ldots, V_{r}$ to the LTS instance on $Y$, one can construct a corresponding supersequence $z_1\cdot z_2 \cdot \cdots \cdot z_r$ for $Y$ and vice-versa.
For this LTS instance the cost of the solution is the length of the corresponding supersequence.

In \autoref{lem:ResourceSchedulingeqLTS}, we assumed a constant number of machines as a premise.
This was to ensure that the instance $\textsc{rs}_{\Pi}$ was polynomial sized in $\Pi$.
The proof of \autoref{lem:ResourceSchedulingeqLTS} shows the corresponding result for LDAG instances in the case that the number of machines $\rho$ is not a constant.
\begin{lemma}\label{lem:ResourceSchedulingeqLTS2}
Let $Y$ be  a conflict-free LDAG.
Consider the LTS problem on the LDAG $Y$.
Let $\rho$ denote the alphabet size of $Y$.
Let $\textsc{rs}_{Y}$ be the corresponding resource scheduling instance. Then the following statements hold.

\textbf{{Part 1:}} Given  a solution  to the LTS  instance on $Y$ of cost $r$, we can find in time polynomial in $|Y| $ and $2^\rho$, a solution to the reduced resource scheduling instance $\textsc{rs}_{Y}$ that achieves makespan $2^\rho \cdot r$.

\textbf{Part 2: } Given  a solution to $\textsc{rs}_{Y}$ achieving  makespan $ \tau$, we can find in time polynomial in $ | Y | $ and $2^\rho$,  a solution to the LTS  instance  on $Y$ of cost (length) $ \tau /2^{\rho -1}$.
\end{lemma}
Applying the same reasoning as in the proof of \autoref{thm:lts-new}, let $Y$ be an arbitrary LDAG, let $\rho$ be the alphabet size of $Y$, and let $\textsc{rs}_{Y}$ be the corresponding resource scheduling problem.
Suppose there is a  $ \beta_Y $-approximation algorithm for $\textsc{rs}_{Y}$ in time polynomial in $|\textsc{rs}_{Y}|$. 
Let $r$ be the length of the minimal supersequence for $Y$.
The first part of \autoref{lem:ResourceSchedulingeqLTS2} implies that there exists a solution to the resource scheduling instance $\textsc{rs}_{\Pi}$ achieving makespan $r \cdot 2^\rho$. 
Hence, the optimal makespan on $\textsc{rs}_{Y}$ is at most $r\cdot 2^\rho $. 
If there exists a polynomial-time $  \beta_Y  $-approximation algorithm \textsc{alg} for the resource scheduling problem, then \textsc{alg} must find a solution to $\textsc{rs}_{Y}$ of makespan at most $ \beta_Y \cdot r \cdot 2^{\rho}$ since the optimal makespan on $\textsc{rs}_{Y}$ is at most $r \cdot 2^\rho$. 

The second part of \autoref{lem:ResourceSchedulingeqLTS} indicates that given a solution to $\textsc{rs}_{Y}$ with makespan $ \beta_Y  r \cdot 2^{\rho}$, we can obtain a solution to LTS instance $Y$ in time polynomial in $ |Y| $ and $2^\rho$ with length $ \beta_Y  r \cdot 2^{\rho} / (2^{\rho-1}) = 2 \beta_Y  r$. 
Since the  optimal solution to $Y$ has length $r$, this means there is a $ 2 \beta_Y  $-approximation algorithm in time polynomial in $ |Y| $ and $2^\rho$.
Thus we have shown the following.
\begin{lemma}\label{resapxtoldag}
Let $Y$ be an LDAG and let $\rho$ be the alphabet size of $Y$.
    If there is a polynomial-time $\beta_Y$-approximation for $\textsc{rs}_{Y}$, then  there is a $ 2 \beta_Y  $-approximation algorithm for the LDAG instance on $Y$ in time polynomial in $ |Y| $ and $2^\rho$.
\end{lemma}

The following definitions, initially defined for sequences,  are extended to LDAGs.

\begin{definition}\cite{LTSnotAPX}
        Let $  X_1, X_2,\ldots,      X_k$ be a collection of LDAGs. Let
$X = X_1 \cdot X_2 \cdot \cdots \cdot  X_k$ denote the LDAG that is obtained by connecting each
$X_i$ to $X_{i+1} $ by a set of directed edges that go from each vertex of
outdegree 0 in $X_i$  to each vertex of indegree 0 in $X_{i+1}$. 
\end{definition}

\begin{definition}\cite{LTSnotAPX}
     Let $\Sigma$ and $\Sigma'$ be two alphabets. Let $a \in \Sigma $ and $b \in \Sigma'$
be two letters. The product $a \times b$ is the composite letter $a \times b \in  \Sigma \times \Sigma'  $.
\end{definition}
\begin{definition} \cite{LTSnotAPX}
The product of an LDAG $X$ and a letter $b$ is the LDAG
 $X \times b$ obtained by taking the product of the label of each vertex of $X$
with $b$. The structure of the LDAG stays the same, only the labels change.
\end{definition}
 
\begin{definition} \cite{LTSnotAPX}
    The product of an LDAG X and a sequence $y  =  b_1 \ldots b_k$
is the LDAG  $X \times y = ( X \times b_1 ) \cdot ( X \times b_2 ) \cdot \dots \cdot  ( X \times b_k )  $.
\end{definition} 
\begin{definition} \cite{LTSnotAPX}
    The product of an LDAG X with a set $Y = \{ y_1, \ldots ,      y_n \} $
of sequences is denoted by  $X \times Y = \cup_{ i=1 }^n  X \times y_i$.
\end{definition}
\begin{definition}\label{prodmach} \cite{LTSnotAPX}
    Let the LDAG X be a union of disjoint sequences, where each sequence is also viewed as a directed path. Then we define the LDAG
$X^k=X ^{k-1} \times X$.
\end{definition}
One can show that for $X$, a restricted SCS instance, $X^k$ is conflict-free.
The number of vertices in the LDAG $X ^ k$ is $n^k$, where $n$ is the number of vertices in  $X$ when viewed as an LDAG.
The alphabet size of the labels of the LDAG $X^k$ is $m^k$, where $m$ is the alphabet size of the labels of the LDAG $X$.

For a set of sequences $X$, the following lemma of \cite{LTSnotAPX} shows a correspondence between supersequences of $X^k$ and supersequences of $X$.
\begin{lemma}\label{factDAG} \cite{LTSnotAPX}
    Let the LDAG $X$ be a disjoint union of sequences and z a
supersequence of $X ^ k$.
We can find $k$ minimal supersequences $z_1,z_2 , \ldots ,     z_k$ of
$X$ such that the product of the length of these sequences is at most the length of
$z$ i.e.  $ |z_1 |  \cdot  | z_2  | \cdot \dots  \cdot | z_k | \leq |z| $. 
Hence, we can find a supersequence of $X $ of size at most  $z ^{ 1 / k } $.
This can be done in time polynomial in  $ | X ^  k |$. 
\end{lemma}

Using \autoref{factDAG}, \cite{LTSnotAPX} obtains the following result.

\begin{lemma}\label{extendapx} \cite{LTSnotAPX}
 For any $k$, given a supersequence $z$ of $X^k$ we can compute a supersequence $z'$ of $X$ of size $|z'| = | z | ^{1 / k }$ in time polynomial in  $X ^ k$.  
  Hence, given a $g(|X|^k) $-time approximation algorithm that achieves an approximation
factor of $f(N)$ where $N$ is the input size whose instance is $X^k$, we can   construct an $f^{  1 / k }(N) $ approximation algorithm for the problem whose instance  is $X$ that runs in time   $  g(|X|^k) + \poly( |X|^ k  )$. 
\end{lemma}
\paragraph{\textbf{The construction}}
Using \autoref{SCSnotapxx2}, let $X$ be an instance of restricted SCS over an alphabet of size 4 that is hard to approximate within factor $c$ in time $O( \poly(|X|) )$, where $c$ is as in \autoref{SCSnotapxx2}. Let $n=|X|$.

 Following \cite{LTSnotAPX}, let $Y=X^k$ 
be an LDAG instance, where $k= 2 \log_c \log  n$. 
 The following result is shown in \cite{LTSnotAPX}.

\begin{lemma}\label{ltstoSCS} \cite{LTSnotAPX}
    If there exists an $  f( |Y| ) $-approximation for the instance $Y$ in time $h(|Y|)$, then there exists a $  f^{1/k}( |Y| ) $-approximation for $X$ in time $O(h(|Y|) +  \poly ( |X|^k )  )$.  
\end{lemma}

Using \autoref{ltstoSCS}, \cite{LTSnotAPX} shows the following result.

\begin{theorem}\label{LDAGapx}\cite{LTSnotAPX}
    There does not exist a $   \log ( |Y| ) $-approximation for the LDAG instance  $ Y$ in time $g(n,k)$  unless  $NP \subset  DTIME( O(  g(n,k) + \poly( n^k )  ) )$.
\end{theorem}

\paragraph{\textbf{Proof of \autoref{thm:lts-jobs}}.}
\begin{proof}
By \autoref{extendapx}, if there were a $f^{ \frac{1}{k} }(|Y|)$-approximation algorithm for the LDAG on $Y$ in time $ O( h(|Y|) $, then there is a $ f(|Y|) $-approximation algorithm for the LDAG problem on $ X $ in time $ O( h(|Y|)  + \poly ( |X|^k )  ) $. 

The resource scheduling instance $\textsc{rs}_{Y}$ has  $| \textsc{rs}_{Y} |  =  2^{4^k} |Y|  =   2^{4^k}  n^k $ jobs.

If there existed an $ o( \log^\alpha | \textsc{rs}_{Y} |  )   $-approximation for $\textsc{rs}_{Y}$ in time $g(n,k)$, then by \autoref{resapxtoldag}, there exists a  $ o(  \log^\alpha | \textsc{rs}_{Y} |  )   $-approximation for $Y$  in time $O( g(n,k) + \poly ( |X|^k )  )$.

We have $4^k  = {   (  c^{ \log_c 4 }  ) }    ^ { 2 \log_c \log n  } =  { (  c^{  \log_c \log n } ) }^  { 2 \log_c 4 } =  (  \log n )^  { 2 \log_c 4 } $.

Observe that 
$$ \begin{array}{ll}
 &     \log^\alpha | \textsc{rs}_{Y} |  =  (           \log  2^{4^k}  n^k     )^ \alpha 
 \\ = &  (           \log  2^{4^k}  + \log  n^k     )^ \alpha  =   (         4^k  \log  2  + k \log  n     )^ \alpha  
 \\ = &   (        (  \log n )^  { 2 \log_c 4 }   \log  2  + k \log  n     )^ \alpha 
   =   O(  ( (  \log n )^  { 2 \log_c 4 }   )  ^\alpha ) .
\end{array}$$

For 
$\alpha =  \frac{1}{2 \log_c 4}  $,  
we have  $  ( (  \log n )^  { 2 \log_c 4 }   )  ^\alpha )  =    \log n $.  
Thus if there is a polytime $o( 
 \log^\alpha ( |V| )  )$-approximation ($|V|$ is the number of jobs) for resource scheduling, then  there exists a $ o( \log(n)  ) $-approximation for $
 Y$  in time $ \poly( 2^{  (  \log n )^  { 2 \log_c 4 } }  n^k ) \leq   O( 2^{  (  \log n )^  { 2 \log_c 4  +1 } }   )    $.
  By \autoref{LDAGapx}, this implies that  $NP \subset DTIME( O( 2^{    \text{polylog} (n)  }  )  ) $.

This shows the theorem.
\end{proof}

\section{Relationship between Scheduling and Shortest Common Super-Sequence (SCS) Problem}
\label{sec:scs}

In this section, we present a conditional lower bound for precedence-constrained resource scheduling problem from the \defn{shortest common supersequence problem (SCS)}~\cite{supseqnotAPX}. 
In SCS, we are given a list $L$ of sequences of integers from alphabet $\Sigma = \{1,2,3,\ldots, \rho\}$ and wish to find a sequence  of minimal length that is a supersequence for each sequence of $L$. 

From~\cite{supseqnotAPX}, the current best approximation ratio for long-standing SCS problem is $|\Sigma|$, even if the size of the \textbf{alphabet $\Sigma$ is a constant}.
We prove the following theorem in this section.
\ResourceImplySCS


\paragraph{\textbf{ Creating a resource scheduling instance from a SCS instance.}}
Given a SCS instance $ \Pi = (\rho,L) $, we construct a resource scheduling instance $\textsc{rs}_{\Pi}$ as  follows.  For each sequence $ \ell\in L$, construct a DAG $D_\ell$ as follows.
 Let $\ell =  \ell_{1},\ell_{2}, \ell_{3},\ldots$. For each character $\ell_i$ in sequence $\ell$, construct a chain $C(\rho, \ell_i)$. Put a directed edge from the sink node of chain $C(\rho, \ell_i)$ to the source node of chain $C(\rho, \ell_{i+1})$. This creates DAG $D_\ell$. The DAG in for resource schedule instance $\textsc{rs}_{\Pi}$ is the union of DAGs $D_{\ell}$ for each $\ell\in L$. All these DAGs $D_\ell$ are independent to each other (i.e., there is no precedence relation between two jobs from different $D_{\ell}$'s).

\begin{figure*}[t]
  \centering
\includegraphics[width=0.75\textwidth]{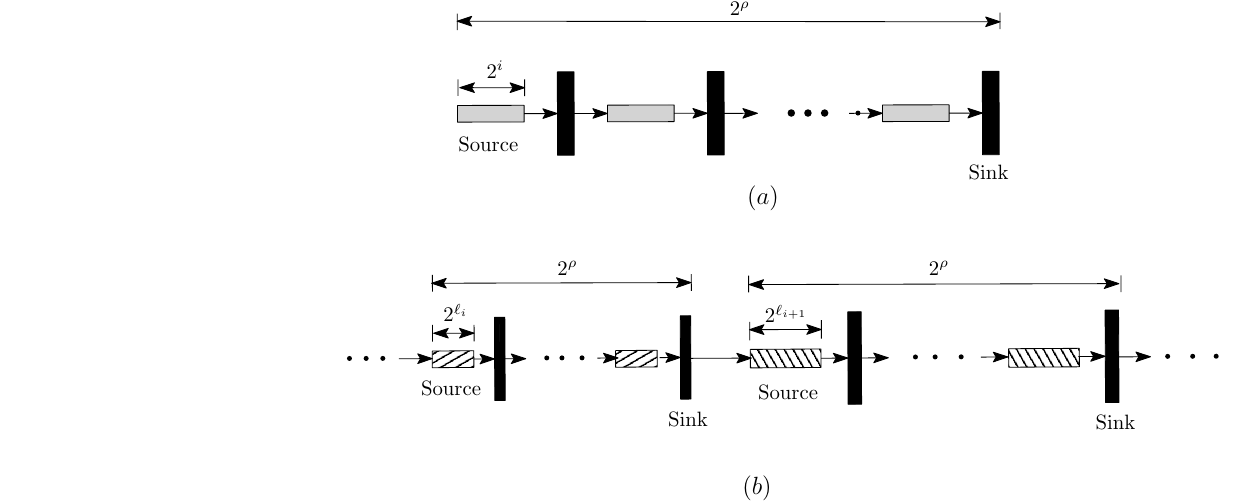}  
   \caption{Part (a) depicts a chain of type $C(\rho, i)$ in $D_{\ell}$ that corresponds to character $i$ in sequence $\ell$. Part (b) shows characters $\ell_{i}$ and $\ell_{i+1}$ from input sequence $\ell = \ldots, \ell_{i}, \ell_{i+1},\ldots$ are replaced with two chains in $D_{\ell}$.}
  \label{fig:SCS-hard}
\end{figure*}

To prove the main theorem of this section, we first prove the following lemma.


\begin{lemma}\label{ResourceSchedulingeqSCS}
Let $\rho \ge 2$ be a fixed constant. Let $\Pi=(\rho, L)$ be a SCS instance where $\rho$ is the size of the alphabet and $L$ is the set of input sequences. Let $\textsc{rs}_{\Pi}$ be the reduced resource scheduling instance. Then the following two statements hold.

   \textbf{{Part 1:}} Given a common supersequence of length $r$ for SCS instance $\Pi$, we can find in polynomial time a solution to $\textsc{rs}_\Pi$ achieving makespan  $2^\rho \cdot r$.
   
   \textbf{{Part 2:}} Given a solution to $\textsc{rs}_\Pi$ achieving makespan $ \tau$, we can find in polynomial time a supersequence for $\Pi$ of length $ \tau /2^{\rho-1}$ 
\end{lemma}

\textbf{Remark:} We do this reduction where $\rho$ is a constant. Hence, $2^\rho$ is still a constant. Thus, if the SCS problem instance has size $X$, the \defn{reduced scheduling problem size is still polynomial in $X$} (in fact, $O(X)$, since $2^{\rho} = O(1)$).

\begin{proof}
    \textit{\textbf{ First part of the lemma (SCS solution $\rightarrow$ scheduling solution).}} Let $z = z_1, x_2,z_3,\ldots, z_r$ be the solution supersequence of length $r$. We claim that the following schedule is a feasible schedule for $\textsc{rs}_\Pi$ that achieves makespan $2^\rho \cdot r$. 

    Given $z$, schedule only ready-to-be-scheduled chains of type $C(\rho, z_1)$ in parallel, followed by only ready-to-be-scheduled chains of type $C(\rho, z_2)$ in parallel, followed by only ready-to-be-scheduled chains of type $C(\rho, z_3)$ in parallel, and so on. 

    Since each input sequence $\ell = (\ell_1, \ell_2, \ldots)$ of SCS instance $\Pi$ is a subsequence of $z$, for each character $\ell_i$, there is a position $i^{'}$ in z, such that $z_{i^{'}} = \ell_i$. Hence, the chain $C(\rho, \ell_i)$ corresponding to $\ell_i$ can be scheduled at the $i^{'}$ epoch of length $2^\rho$.

    \paragraph{\textbf{ Second part of the lemma (scheduling solution $\rightarrow$ SCS solution).}}
\myworries{ for resource scheduling -> for the resource scheduling? }
    Let $\mathcal{A}$ be an algorithm for resource scheduling problem that gives a schedule $S$ with makespan $\tau$ on the reduced resource scheduling instance $\textsc{rs}_{\Pi}$. Note that instance $\textsc{rs}_{\Pi}$ can be thought of as a DAG where each node $u$ is a chain $C_u$, and if there is an edge from node $u$ to $v$, then chain $C_u \prec C_v$. From \hao{changed}  Lemma~\ref{cleanschedule}, we can find an algorithm $\mathcal{A}$ that gives a schedule $S^{'}$ with makespan at most $2 \cdot \tau$ with the following two properties:

(1) algorithm $\mathcal{A}'$ does not start any job while another job has already started but not finished, and 

(2) at any point in time, $\mathcal{A}'$ runs only one type of skinny job (i.e., all the skinny jobs that run at the same time have the same length).

\paragraph{\textbf{ Creating a new schedule $S^{''}$ where jobs of a chain run without any gap.}}
Let $\mathcal{S}^{'}$ be the schedule by algorithm $\mathcal{A}^{'}$. Let $t_{i,j}$ be the timestep in schedule $\mathcal{S}^{'}$ when skinny-job of length $2^i$ is scheduled $j\cdot (2^\rho/2^i)$ times (i.e., the sum of lengths of skinny-jobs each of length $2^i$ until time $t_{i,j} + 2^i$ is $j\cdot 2^\rho$). \myworries{ -> are scheduled for the $j\cdot (2^\rho/2^i)$-th time to clarify its at the start? } Note that $t_{i,1} < t_{i,2} < t_{i,3}< \cdots$ for each skinny-jobs of length $2^i$. Besides, given two skinny jobs of lengths $2^{i}$ and $2^{i^{'}}$, and two integers $j$ and $j^{'}$, there is an ordering of $t_{i,j}$ and $t_{i^{'}, j^{'}}$ (i.e., either $t_{i,j} < t_{i^{'}, j^{'}}$ or $t_{i^{'}, j^{'}} < t_{i,j}$) in schedule $\mathcal{S}^{'}$. Since there is a total ordering of all the $t_{i,j}$'s, create a sequence $X$ of integers $i$ in the order of $t_{i,j}$'s (e.g., if $t_{5,1} < t_{3,1} < t_{5,2} < t_{5,3} < t_{4,1} < t_{3,2} < \cdots$, then the sequence $X$ of integers is $(5,3,5,5,4,3,\ldots)$, where $X[1] = 5, X[2] = 3, X[3] =5, X[4] = 5$ etc.). Create a new schedule $\mathcal{S}^{''}$ from $X$ as follows: Schedule skinny jobs of length $2^{[X[1]]}$ for $2^\rho/2^{[X[1]]}$ times (i.e., their total length is $2^\rho$; call this time interval as 1st epoch), followed by skinny jobs of length $2^{[X[2]]}$ for $2^\rho/2^{[X[2]]}$ times (call this time interval as 2nd epoch),  followed by skinny jobs of length $2^{[X[3]]}$ for $2^\rho/2^{[X[3]]}$ times (call this time interval as 3rd epoch), and so on.  
\haonew{added}
From the construction of $S^{''}$ it is clear that $S^{''}$ obeys the precedence constraints.

\paragraph{\textbf{ The length of schedule $S^{''}$ is no more than $S^{'}$.}}
The schedule $S^{''}$ is constructed by repositioning the skinny jobs from $S^{'}$. Hence, the length of $S^{''}$ is no more than $S^{'}$. 

\paragraph{\textbf{ Each chain $C$ is run on schedule $S^{''}$.}}
Recall that from construction, each $D_{\ell}$ is a concatenation of chains $C$. For each chain $C$ in $D_{\ell}$, we show that $C$ is scheduled in $S^{''}$. 
Let $C$ be a chain of type $C(\rho,i)$. Let in schedule $S^{'}$, chain $C$ starts at time $t_s(C)$ and finishes at time $t_f(C)$ (since, $S^{'}$ is a feasible schedule, it must schedule $C$). Then there must be a $t_{i,j} \in [t_s(C), t_f(C) -1 ]$, since chain $C$ of type $C(\rho,i)$ starts and finishes in $[t_s(C), t_f(C)]$. Find the index $k_C$ in the integer-sequence $X$ that corresponds to $t_{i_C,j_C}$. Hence, $X[k_C] = i_C$. 

We claim that for any chain $C$ its predecessor is completed by the $k_C-1$-th epoch.
We prove by this induction on the number of ancestors of $C$.
If $C$ has no predecessor, then the statement is clear.
Let $C'$ be the predecessor of $C$.
Then by induction the predecessors of $C'$ are complete after the $ k_{C'}-1 $-th epoch. So $C'$ will be completed by the $ k_{C'} $-th epoch. 
Since $C'$ must finish before $C$ starts in schedule $S'$,
$ t_{i_{C'},j_{C'} } \leq t_f(C') -1 < t_s(C) \leq  t_{i_C,j_C} $. 
Thus $ k_{C'}  \leq k_C -1$. 
This means that each $C$ will be completed by the $k_C$-th epoch. Since $k_C$ exists for all $C$, this means the schedule finishes.
So $C'$ completes by the $k_C-1$-th epoch, thus completing the induction.

\end{proof}

We now prove the main theorem (conditional lower bound for resource scheduling problem) of this section.

\paragraph{\textbf{ Proof of Theorem~\ref{thm:scs}.}}
\begin{proof}
Lemma~\ref{ResourceSchedulingeqSCS}  gives a reduction from  finding a $c$-approximation for the SCS problem to finding a $c/2$-approximation resource scheduling as follows.

Given a SCS instance $ \Pi = (\rho,L) $ with optimal supersequence length $r$, the first part of \autoref{ResourceSchedulingeqSCS} implies that there exists a solution to resource scheduling on $\textsc{rs}_{\Pi}$ using time $2^\rho \cdot r$. Hence, the optimal makespan on $\textsc{rs}_{\Pi}$ is at most $2^\rho \cdot r$. 
If there exists a polynomial-time $c/2$-approximation algorithm \textsc{alg} for the resource scheduling problem, then \textsc{alg} must find a solution to $\textsc{rs}_{\Pi}$ of makespan at most $ c \cdot r\cdot 2^{\rho-1}$, since the optimal makespan on $G_L$ is at most $2^\rho \cdot r$. 

The second part of \autoref{ResourceSchedulingeqSCS} implies that given a solution to $\textsc{rs}_{\Pi}$ with makespan $2^\rho r$, we can obtain a solution to $ \Pi = (\rho,L) $ in polynomial-time with supersequence length $c \cdot r$. Hence, we can find a $c$-approximate solution to $ \Pi = (n,L) $ in polynomial time (by first, reducing the SCS instance $\Pi = (n,L)$ to a resource scheduling instance $\textsc{rs}_{\Pi}$, then, finding a $\rho/2$-approximate solution to $\textsc{rs}_{\Pi}$, and then converting the solution to $\textsc{rs}_{\Pi}$ to a solution to $\Pi$---all three steps are done in polynomial time). The theorem is thus proved.
\end{proof}

\section{Tight Lower Bound of the Online Problem}
\label{sec:lower-online}
In this section, we prove the following lower bound on the competitive ratio of any randomized online algorithm. The first lower bound is dependent on the number of jobs, thus it holds even if there is only one type of resource. The second lower bound is dependent on the number of resource types.

\begin{figure*}
  \begin{center}
    \includegraphics[width=0.75\textwidth]{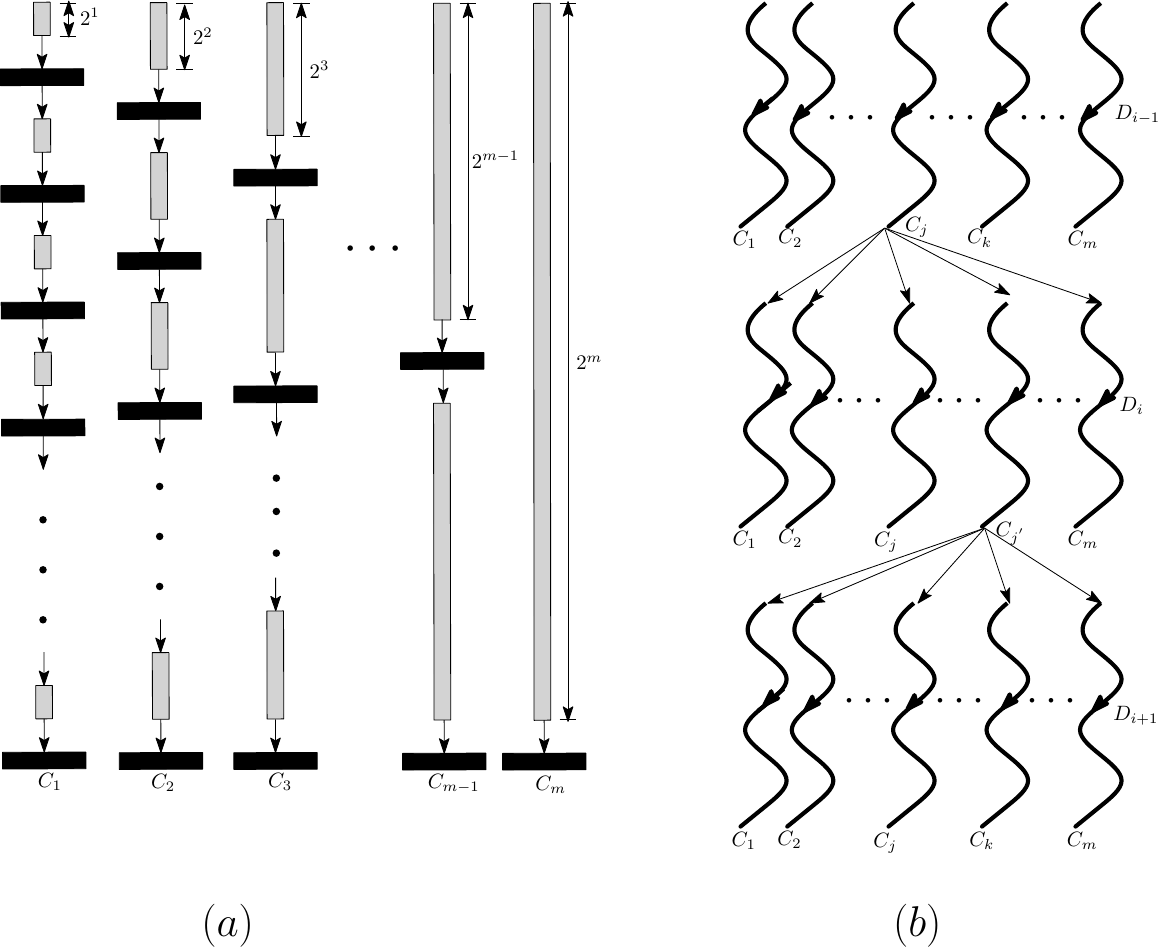}
  \end{center}
  \caption{Part (a) shows a $D_i$. Part (b) shows how $D_i$'s are connected to each other. In $D_{i-1}$, chain $C_j$ contains end($i-1$). In $D_{i}$, chain $C_{j^{'}}$ contains end($i$).}
  \label{fig:online-lower}
\end{figure*}

\subsection{Lower Bound as a Function of the Number of Jobs}
\OnlineRSHard

We prove Theorem~\ref{Thm:online-lower} using Yao's minimax principle~\cite{motwani1995randomized}. 
\begin{lemma}[Yao's minimax principle~\cite{motwani1995randomized}]
The optimal deterministic performance for an arbitrary input distribution \( D \) serves as a lower bound for the performance of any randomized algorithm \( R \) when faced with its worst-case input.
\end{lemma}

For this, we first construct a randomized DAG as follows.

\paragraph{\textbf{ Constructing a randomized instance of the resource scheduling problem.}}
Let $m$ be an arbitrary positive integer. It will turn out that $m = \theta(\log n)$ where $n$ is the number of jobs. 
\hao{specify n is the number of jobs we wish our instance to have? $n$ will have the form} We construct a randomized DAG $D=(V,E)$ as follows. DAG $D$ has a series of gadgets $D_i$ for $i\in [1, 2^m]$. Each $D_i$ is a collection of $m$ chains where the $j\in [1,m]$-th chain is $C(m,j)$. Chains in each $D_i$ are independent, i.e., there is no precedence relation between two jobs from different chains of the same $D_i$. However, chains from two different $D_i$'s are not independent. The precedence relation between the chains from different $D_i$'s are defined below.\haonew{m is can be an arbitrary positive integer the number of nodes n will depend on m. It turns out that $ m=\theta (\log n) $  }

Take any two consecutive gadgets $D_i$ and $D_{i+1}$. Only one chain $C_i$ from $D_i$ has precedence over all the chains in $D_{i+1}$. Call the chain $C_i$ as \defn{blocking-chain} for gadget $D_{i+1}$, if $C_i$ has precedence over all the chains in $D_{i+1}$. In the randomized DAG $D$, for each chain $C_i$ in $D_i$,
$$
\Pr[C_i \text{ is the blocking chain for gadget } D_{i+1}] = \frac{1}{m}.
$$

If chain $C_j \in D_i$ is randomly chosen as the blocking-chain, then create a directed edge from the sink node $C_j$ to the source node of every chain in gadget $D_{i+1}$. Call the sink node $C_j$ as $\texttt{end}(i)$.

We first prove the following lemmas before proving the main theorem of this section.

\begin{lemma}\label{completechain}
Let $D$ be the randomized DAG in the resource scheduling problem instance. Let $\mathcal{A}$ be any deterministic online algorithm for the resource scheduling problem. Then,
\vspace{-0.5em}
\begin{multline*}
\E[\text{Number of chains in gadget } D_i \text{ that finish before or at the same} \\ \text{ time with }\texttt{end}(i)] \geq \frac{m}{4}.
\end{multline*}

\end{lemma}
\begin{proof}
At any time when fewer than $\frac{m}{2}$ chains have finished,  the probability  that $\texttt{end}(i)$ is not among the endpoints of the completed chains is at least 1/2.
Thus the probability that $\texttt{end}(i)$ is finished before $\frac{m}{2}$ chains have completed is at most 1/2.
Thus when $\texttt{end}(i)$ is completed, the  probability that at least  half the chains have been completed at this time is at least 1/2.
It follows that the expected number of completed chains is at least $\frac{m}{4}$.
\end{proof}

\begin{lemma}
\label{lem:onl-a}
Let $\mathcal{A}$ be any deterministic online algorithm and let $T_{\mathcal{A}}$ denote the makespan of $\mathcal{A}$ on randomized DAG $D$. Then, $\E[T_{\mathcal{A}}] \ge \Omega(m\cdot 4^m).$
\end{lemma}
\begin{proof}
Let $T_{\mathcal{A}}$ be the random variable that denotes the makespan of algorithm $\mathcal{A}$ on randomized DAG $D$. For $i\in [2,2^m]$, let $T_{D_{i}}$ be the random variable that denotes the time taken by algorithm $\mathcal{A}$ from running $\texttt{end}(i-1)$ until running $\texttt{end}(i)$. Then,
$
T_{\mathcal{A}} \ge T_{D_{2}} + T_{D_{3}} + \cdots + T_{D_{4^m}}.
$
Using linearity of expectation, the expected makespan of algorithm $A$ follows.
$$
\E[T_{\mathcal{A}}] \ge \E[T_{D_{2}}] + \E[T_{D_{3}}] + \cdots + \E[T_{D_{4^m}}].
$$

From Lemma~\ref{completechain}, we know $\E[T_{D_{j}}] \ge (m/4) \cdot 2^m$ for each gadget $j\in [1,2^m]$. Hence,
$
\E[T_{\mathcal{A}}] \ge \Omega(m\cdot 4^m).
$
The lemma is thus proved.
\end{proof}

\begin{lemma}\label{lem:offline-a}
Let $T_{\textsc{opt}}$ be the makespan of the optimal offline algorithm OPT. Then,
$T_{\textsc{opt}} =  O(4^m + m\cdot 2^m).$
\end{lemma}
\begin{proof}
An optimal offline algorithm OPT knows all the blocking-chains in each gadget $D_i$. OPT runs the blocking chains in each $D_i$ one after another. This takes total time $O(2^m \cdot 2^m) = O(4^m)$, since there are $2^n$ gadgets. After that, there are no precedence constraints between any two remaining chains in $D$ (no precedence between any two jobs from two different chains, although two jobs in the same chain are not independent).

After each $C(m,i)$ for $i=1,2,.., 2^m$ is available, OPT processes all remaining chains in $D$ as follows: run only the chains $C(m,1)$ in all the $D_i$'s in parallel, followed by only the chains $C(m,2)$ in all the $D_i$'s in parallel, followed by only the chains $C(m,3)$ in all the $D_i$'s in parallel, and so on. From \hao{changed} Property~\ref{skinnyjobpar}, each independent set of $C(m,j)$ for $j\in  [1,m]$ takes $2^m$ units of time. Hence, OPT completes $D$ by time $O(4^m + m\cdot 2^m)$.
\end{proof}

\paragraph{\textbf{ Proof of Theorem~\ref{Thm:online-lower}}}
\begin{proof}
From Lemma~\ref{lem:offline-a}, we know that the optimal offline makespan $T_{\textsc{opt}} =  O(4^m + m\cdot 2^m).$ From Lemma~\ref{lem:onl-a}, we know that for any deterministic online algorithm $\mathcal{A}$, $\E[T_{\mathcal{A}}] \ge \Omega(m\cdot 4^m)$. Hence, 
$$
\frac{\E[T_{\mathcal{A}}]}{T_{\textsc{opt}}} \ge \Omega(m).
$$
Using Yao's minimax principle, the lower bound on the competitive ratio of any randomized algorithm is thus $\Omega(m)$.

Note each $D_i$ contains $O(2^m)$ nodes.  
So the number of nodes $n$ in $D$ is $O(4^m)$ and the maximum length $t_{\max}$ of a job is $2^m$. 
Hence, $m = \Theta(\log n)$ and $m =\log t_{\max}$. The theorem is thus proved. 

\end{proof}

\subsection{Lower Bound as a Function of Number of Resource Types}


\thmonlinelowermult

We prove Theorem~\ref{thm:online-lower-multi} using Yao's minimax principle~\cite{motwani1995randomized}. 

For this, we first construct a randomized DAG as follows.

\paragraph{\textbf{ Input randomized DAG construction.} } We construct our DAG as follows. We have $d$ layers $L_1,L_2,.., L_d$ of nodes, where $d$ is the number of resource types. Let $m$ be an arbitrary positive integer. We will end up choosing $m=O(d^2)$ to get the right competitive ratio. 
There are $m$ nodes in each layer. 
    Each node in layer $L_i$ is a job that takes 1 unit of time and requires one unit of resource $\mathcal{R}_i$ (the job has no resource requirement for other types of resource except resource $\mathcal{R}_i$). 
Jobs in each layer $L_i$ are independent, i.e., there is no precedence relation between two jobs in the same $L_i$. However, jobs from two different $L_i$'s are not independent. The precedence relation between the jobs from different $L_i$'s are defined below.

Take any two consecutive layers $L_i$ and $L_{i+1}$. Only one job from $L_i$ has precedence over all the jobs in $L_{i+1}$. Call the job $j\in L_i$ as \defn{blocking-job} for gadget $L_{i+1}$, if $j$ has precedence over all the jobs in $L_{i+1}$. In the randomized DAG $D$, for each  job $j\in L_i$,
$$
\Pr[j \text{ is the blocking job for layer } L_{i+1}] = \frac{1}{m}.
$$

If job $j \in L_i$ is randomly chosen as the blocking job, then create a directed edge from job $j$ to every job in layer $L_{i+1}$. Denote blocking job $j\in L_i$ as $b_i$.

The online algorithm does not know about any nodes in layer $L_{i+1}$ until the algorithm finishes the blocking node $b_i$. The complexity of the online setting thus comes from the fact that the online algorithm does not know which node in layer $L_{i}$ is the blocking node for the next layer $L_{i+1}$. 

The optimal offline algorithm $\textsc{opt}$, however, knows the entire DAG offline and also the blocking nodes in each layer. 
Let $T_{\textsc{opt}}$ denote the makespan of $\textsc{opt}$.

\begin{lemma}
Let $\mathcal{A}$ be any deterministic online algorithm for resource scheduling. Let $T_{\mathcal{A}}$ be the makespan of $\mathcal{A}$ on randomized input DAG $D$. Then, 
$$
\frac{\E[T_{\mathcal{A}}]}{T_{\textsc{opt}}} \ge \frac{(d-1)m}{d+m}.  
$$
\end{lemma}

\begin{proof}
Define $ t_i(\mathcal{A}) $ to be the time algorithm $\mathcal{A}$ completes the blocking job $b_{i-1}$ for $2\le i \le d$. 
Since $ \mathcal{A} $ does not know which job in layer $L_{i-1}$ is the blocking job $b_{i-1}$ for layer $L_i$, in expectation $ t_i(\mathcal{A}) - t_{i-1}(\mathcal{A}) $ is at least  $ m/2$.
Thus in expectation $  t_{d-1}(\mathcal{A}) $ is at least $ (d-1)m $ and so is a lower bound of makespan of  $\mathcal{A}$.
    
The optimal offline algorithm \opt, however, knows all the blocking jobs and could first complete $b_1,b_2,.., b_{d-1}$ in $d-1$ time-steps. After that, \opt takes one job from each of the $d$ layers and run them in parallel. Notice that a job in layer $L_i$ only requires resource $\mathcal{R}_i$, whereas a job from a different layer $L_j$ requires only resource $\mathcal{R}_j$, where $\mathcal{R}_i \neq \mathcal{R}_j$. Hence, two jobs from different layers can be run in parallel. But two jobs from the same layer can not be run in parallel, since each job in layer $L_i$ requires the entire resource capacity of $\mathcal{R}_i$. 

Thus the total time \opt takes is at most $(d + m)$. This proves a $\frac{(d-1)m}{d+m} $ 
lower bound on the expected competitive ratio for any deterministic online algorithm on a randomized input DAG instance.

\end{proof}
Choosing $m>3d^2$  in the previous lemma shows that no algorithm can be $\frac{d-1}{2}$-competitive. \hao{added needs rewording?}

Using Yao's minimax principle, the lower bound on the competitive ratio of any randomized algorithm is thus $\Omega(d)$.
\section{Online Algorithm with Tight Competitive Ratio}\label{OnlineAlgSection}

In this section, we present a competitive online algorithm $\onl$.

Recall \hao{Probably don't need to say here since it's the last section probably recall?} we are given a DAG $D = (V,E)$ and $d \ge 1$ types of resource where resource $\mathcal{R}_i$ has capacity $r_i$ (after normalizing separately for each resource type, assume that each resource type has capacity $r$).
Each node $v\in V$  is a job that requires $r_{i,v}$ units of resource $\mathcal{R}_i$ for $t_v$  units of time to complete.   
For each edge $(u,v) \in E$,  we must complete job $u$ before starting job $v$. 
At any time, for each resource $\mathcal{R}_i$, we may use a total of $r$ units of resource $\mathcal{R}_i$.  \haonew{did we not rescale so there is one unit of each resource?}
The objective is to complete all jobs as soon as possible. \hao{this also needs to be in the first section?}

By losing at most a factor 2, we may round up all processing times $t_v$ to the nearest power of 2. 
Henceforth, we assume all processing times are powers of 2. Also, assume that there is a single source and a single sink node in DAG $D$, otherwise, we add a dummy source and a dummy sink node to $D$.

Our main theorem of this section is as follows.
\OnlineLogApx

Before presenting the online algorithm, we define the following useful notations.

In the online setting, if all the predecessors of a job have finished their execution, the job is revealed to the online algorithm. Call a job a \defn{ready job}, if the job is not executed yet, but all of its predecessors have finished execution. For each ready job $v$, algorithm $\onl$ assigns its depth $\delta(v)$ and level $\psi(v)$ as follows. If $v$ is the source node, then $\delta(v) = \psi(v) = 1$. For all other nodes,
$$
\delta(v) = \max\{\delta(p) + t_p | \text{ where } p \text{ is a parent of node } v\}. $$

Let the length of job $v$ be $2^i$. Level $\psi(v)$ is assigned as the smallest integer multiple of $2^i$ which is at least as large as $(\psi(p) + t_p)$ for all parent $p$ of job $v$. That is, 
\begin{equation*}
 \begin{split}
\psi(v) = \min\{k\cdot 2^i | k \in \N \text{ and }  k\cdot 2^i \ge \psi(p) + t_p  \text{ for all of } v \text{'s} \\ \text{ predecessors } 
 p\}.
 \end{split}
 \end{equation*}

\paragraph{\textbf{ Online algorithm ONL.}} Our online algorithm $\onl$ maintains a set, called \defn{Ready-Jobs} that is the set of ready jobs. Although, any job in the set Ready-Jobs can be executed if enough resource is available for the ready job, algorithm \onl instead execute the ready jobs in a \defn{time-indexed} fashion as follows: 

\onl executes ready jobs in stages---in each stage, \onl executes only the ready jobs that are assigned the same level. If all the ready jobs that are assigned level $\ell$ have finished execution, \onl finds the next smallest level $\ell^{'} > \ell$ such that there are some ready job(s) that are assigned level $\ell^{'}$. \onl executes only the ready jobs that are assigned level $\ell^{'}$ (i.e., there could be other ready jobs sitting idle because their levels are higher than $\ell^{'}$). After executing level $\ell^{'}$ jobs, new jobs might be revealed and they become ready jobs. We insert these new ready jobs into the set Ready-Jobs. Find the next level $\ell^{''} > \ell^{'}$, such that there is at least one ready job with level $\ell^{''}$. Execute all the jobs with level $\ell^{''}$ and then find the next non-empty level, execute those jobs, and so on. 

The algorithm finishes when there is no more ready jobs.


Define the set of jobs that are assigned level $i$, as \defn{Working-Set$(i)$}.

\SetKwComment{Comment}{/* }{ */}

\begin{algorithm}[hbt!]
\caption{Online Algorithm \onl}\label{alg:two}
Ready-Jobs = $\{s\}$       \Comment*[r] {A job is ready if it is not yet executed but all its predecessors are already executed. $s$ is the source node of the DAG.}

Current-Time = 0

\While{Ready-Jobs is not empty}{
    $i$ = getNextNonEmptyLevel(Ready-Jobs)    \Comment*[r] {Find the next level $i$ such that there is a ready job that is assigned level $i$.}
    
    Working-Set$(i)$ = getJobs(Ready-Jobs, $i$)   \Comment*[r] {Working-Set$(i)$ is the set of ready jobs which are assigned level $i$.}
    
    Processing-Time = execute(Working-Set$(i)$)    \Comment*[r] {Run the jobs in Working-Set$(i)$. The function is described below.}
    
    New-Ready-Jobs = getNewReadyJobs(Ready-Jobs, Working-Set$(i)$) \Comment*[r] {Find the set of new ready jobs that reveal after executing Working-Set$(i)$.}
    
    assignLevel(New-Ready-Jobs)     \Comment*[r] {Assign level $\psi$ for each new ready job.}
    
    Ready-Jobs = ReadyJobs $\backslash$ Working-Set$(i)$ $\cup$ New-Ready-Jobs
    
    Current-Time = Current-Time + Processing-Time
    
}

Makespan = Current-Time
\end{algorithm}

\begin{lemma}\label{lvlleq2depth}
Let $\hat{\delta}$ and $\hat{\psi}$ be the depth and level assignment of the sink node of the input DAG. Then, $\hat{\psi} < 2\cdot \hat{\delta}.$
\end{lemma}
We defer the proof to the Appendix~\ref{sec:appendix}.

\paragraph{\textbf{ Function execute(Working-Set$(i)$).}}

All the jobs in Working-Set$(i)$ have the same level assignment, and thus they are independent (no precedence among them)---If job $u$ is a predecessor of job $v$, then $\psi(u) < \psi(v)$.

ONL greedily starts executing the jobs in Working-Set$(i)$ in parallel until no more jobs can be started due to resource budget for any of the resource $\mathcal{R}_i$. 
Whenever some job is finished, check if any job that has not started yet can be executed obeying the total resource budget $r$ for each resource $\mathcal{R}_i$. 
That is, at any time execute runs a maximal set of the remaining jobs whose total resource cost for each resource $\mathcal{R}_i$ does not exceed the budget $r$.

Define \defn{$\text{cap}_i$} as the largest power of 2 that divides $i$ (e.g., $\text{cap}_{32} = 32, \text{cap}_{33} = 1, \text{cap}_{28} = 4, \text{cap}_{38} = 2$ ). 

For a set of jobs $W_i$ and resource $\mathcal{R}_j$, define $\work_j(W_i) = \sum_{v \in W_i} r_{i,v} \times t_v$.

\begin{lemma}\label{maxjobassigned}
Let job $v$ of length $t_v$ be assigned to level $i$. Then, $t_v \le \texttt{cap}_i$. 
\end{lemma}

\begin{proof}
    By assumption, all the job's sizes are powers of 2. Hence, $t_v$ is some power of 2, say $2^j$ and its assigned level $i$ is some multiple of $2^j$, say $k \cdot2^j$.  Then $\texttt{cap}_{k \cdot 2^j}  \geq \texttt{cap}_{ 2^j} =2^j =t_v $. 
\end{proof}

\begin{lemma}\label{CompleteLvlTime}
Let $W_i \neq \emptyset$ be the set of jobs that are assigned to level $i$. Let $\tau_i$ be the time it takes by execute($W_i$) to finish all the jobs in $W_i$. Then,
$$
\tau_i \leq 2 \sum_{j=1}^d \work_j(W_i)/r + \min\{t_{\max}, \texttt{cap}_i \}.
$$
\end{lemma}
\begin{proof}  \myworries{move proof to appendix?} 
Let $t$ be a time-step in the time-period $\tau_i$. Call time-step $t$ \defn{resource efficient}, if for some resource $\mathcal{R}_j$, the total resource usage of resource $\mathcal{R}_j$ of all jobs being processed at $t$ is at least $r/2$. Otherwise, call time-step $t$ \defn{resource inefficient}. 

We show that $\tau_i$ has at most $ \max_{v \in W_i}  t_v \leq  \min\{t_{\max}, \texttt{cap}_i \} $ resource-inefficient time-steps.
Let $t'$ be the first resource-inefficient time-step in $\tau_i$. Hence, for each resource $\mathcal{R}_j$, there is at least $r/2$ units of $\mathcal{R}_j$ not allocated at time -step $t$. Since execute(Working-Set($i$)) greedily schedules jobs, it must be that all jobs in $W_i$ of resource usage less than $r/2$ have already been scheduled.
Thus after all jobs that are running at $t'$ have finished, there will be no more resource-inefficient steps.
Starting from time $t'$, it can take at most $ \max_{v \in W_i}  t_v $ timesteps for all jobs started at $t'$ to have finished.
Thus there are at most $ \max_{v \in W_i}  t_v $ resource-inefficient time-steps.

There can be at most $2 \sum_{j=1}^d \work_j(W_i)/r$ efficient time-steps. Since, in a resource-efficient time-step, at least one type of resource, say $\mathcal{R}_j$ is efficiently being used, i.e., at least $r/2$ units of $\mathcal{R}_j$ is allocated by jobs being processed at $t$. If $\mathcal{R}_j$ is efficiently being used at time-step $t$, then time-step $t$ contributes at least $r/2$ in the quantity $\work_j(W_i)$. Hence, there are at most $2\cdot\work_j(W_i)/r$ resource-efficient time-steps where $\mathcal{R}_j$ is efficiently being used. There are $d$ types of resources, hence, the number of resource-efficient time-steps is at most $\sum_{j=1}^d \work_j(W_i)/r$. 

From~\autoref{maxjobassigned}, \hao{added lemma number changed to autoref } we know that $\max_{v \in W_i}  t_v  \le \texttt{cap}_i$. Also, we have $\max_{v \in W_i}  t_v \le t_{\max}$. Summing all the resource-efficient and resource-inefficient time-steps, we get $\tau_i \leq 2 \sum_{j=1}^d \work_j(W_i)/r + \min\{t_{\max}, \texttt{cap}_i \}$, and the lemma is thus proved.
\end{proof}

From the execution of algorithm \onl, the following lemma is immediate.

\begin{lemma}\label{SummingUp}
Let $T_{\onl}$ be the makespan of our online algorithm. Let $\psi^{-1}(i)$ denote the set of jobs that are assigned level $i$. Then, 
$$
T_{\onl} \leq \sum_{i = 0}^{\ell_{\max}} \tau_i \cdot 1_{ \psi^{-1}(i) \neq \emptyset }.
$$
\end{lemma}

To prove an $O(\log |V|)$-approximation, we first prove the following identity before proving the makespan upper bound of algorithm \onl.
\begin{lemma}\label{SumIdentity}
Let $\ell_{\max}$ be the maximum level assignment to any job by algorithm \onl. For any $ \hat{\ell} \leq \ell_{\max} $,
$\sum_{\substack{i=1:  \\ \psi^{-1}(i) \neq \emptyset } }^{\hat{\ell}}   \texttt{cap}_i  \leq 
(\log |V| )  \hat{\ell}  $, where $V$ is the set of jobs in the input DAG.
  
\end{lemma}

\begin{proof} \myworries{move to appendix?}
If algorithm \onl assigns a job of size $2^m$ to level $i$, then $2^m \le \texttt{cap}_i$, Let $R_m$ be the number of levels $i\le \hat{\ell}$ such that $\texttt{cap}_i = 2^m$.  There are $s = \lceil\log \hat{\ell}\rceil$ different powers of 2 sizes of jobs possible whose level assignments are at most $\hat{\ell}$, since $t_v \le \psi(v)$ for any job $v$ of size $t_v$.  

 From the definitions of $R_m$ and $\texttt{cap}_i$, the following holds.
 $$\sum_{\substack{i=1:  \\ \psi^{-1}(i) \neq \emptyset } }^{\hat{\ell}} \texttt{cap}_i = \sum_{m=1}^{s} 2^m \cdot R_m.$$ 
 Since, we want an upper bound, we have the following optimization problem.

\begin{equation}
 \text{Objective function:  \;\;\;\;\;\;} \max \sum_{m=0}^{s} 2^m \cdot R_m 
\end{equation}

 Subject to the following two constraints,

 \begin{equation}
 0 \leq R_m \leq \hat{\ell} /2^{m-1} \leq 2^{s-m+1}
 \end{equation}

 \begin{equation}
 \sum_{m=0}^{s} R_m\leq |V|
 \end{equation}

 We use the following remark for the rest of the proof.
 \begin{remark} \label{fracknapsack}  
For $0 \le c_0 < c_1 < \cdots < c_t $,  $ d_0,d_1,..,d_t \geq 0 $ the maximum of the optimization function  $ \max \sum_{i=0}^t  c_i \cdot x_i $, subject to constraints $ \sum_{i=0}^t x_i \leq M, \  \ d_i \geq x \geq 0 $ is achieved by finding a minimal $\ell$ such that we can make $x_i=d_i$ for all $t - \ell \le i \le  t$ and setting $x_{t - \ell-1} = M-  \sum_{i=t - \ell}^t x_i$.
 \end{remark}

By \autoref{fracknapsack}, the maximum value of the objective function subject to the constraints can be obtained by finding a minimal $\ell \leq s$ such that we can make each of $R_m$ from $s \ge m \ge (s - \ell)$ takes its maximum value $2^{s-m+1}$ and setting $R_{s-\ell-1} = |V| - \sum_{m = s - \ell}^s 2^{s -m+1} $.

From the definition of $\ell$, we have
$$
\sum_{m = s - \ell}^s 2^{s -m+1}  = \sum_{p = 1}^{\ell+1} 2^p \le |V|.
$$

Hence, $\ell = O(\log |V|)$.

Thus we get the following upper bound of the objective function.

$$
\sum_{m=0}^{s} 2^m \cdot R_m \le 2^{s-\ell -1} \cdot 2^{s - (s -\ell -1)+1} + \sum_{m = s - \ell}^s 2^m \cdot 2^{s -m +1}   = 2^{s+1} + \sum_{m = s - \ell}^s 2^{s+1} \le 2\cdot \ell \cdot 2^{s+1} = O(\hat{\ell} \cdot \log |V|).
$$
The lemma is thus proved.
\end{proof}

We now present two upper bounds in Lemma~\ref{lower1} and Lemma~\ref{bdInTermstmax} of the makespan of our online algorithm as follows.
\begin{lemma}
\label{lower1}
Let $V$ be the set of jobs in the input DAG. Let $\ell_{\max}$ be the level assignment to the sink node. Then, the upper bound of makespan $T_{\onl}$ of online algorithm $\onl$ is as follows.
$$
T_{\onl} \leq 2\sum_{j=1}^d \work_j(V)/r + O(\ell_{\max} \cdot \log |V|).
$$
\end{lemma}
\begin{proof}
    By Lemma~\ref{CompleteLvlTime} and Lemma~\ref{SummingUp}, $T \leq 2\sum_{j=1}^d \work_j(V)/r  + \sum_{i=0}^{\ell_{\max}} \min\{t_{\max}, \texttt{cap}_i\}$. 

By \autoref{SumIdentity} it follows that $T \leq 2\work(V)/r + O(\ell_{\max} \cdot \log |V|)$.
\end{proof}

\begin{lemma}\label{bdInTermstmax}
Let $V$ be the set of vertices (jobs) in the input DAG with total resource capacity $r$. Let $t_{\max}$ be the maximum job size and $\ell_{\max}$ be the maximum level a job is assigned by algorithm \onl. Then, the upper bound of makespan $T_{\onl}$ of algorithm \onl is as follows.

$$
T_{\onl} \leq 2\sum_{j=1}^d \work_j(V)/r + O(\ell_{\max} \cdot \log t_{\max}).
$$
\end{lemma}

\begin{proof}
Note that the number of levels  from $ 1 $ to $\hat{\ell}$  such that $ \texttt{cap}_i = 2^m$ is $\lfloor \hat{\ell}/2^m  \rfloor -\lfloor \hat{\ell} /2^{m+1} \rfloor$.
Then $\sum_{\substack{i=0:   } }^{\hat{\ell}}   \texttt{cap}_i =  \sum_{m=0}^n 2^m ( \lfloor \hat{\ell}/2^m  \rfloor -\lfloor \hat{\ell} /2^{m+1} \rfloor )  =   \sum_{m=0}^{\log \hat{\ell}} 2^m ( \lfloor \hat{\ell}/2^m  \rfloor -\lfloor \hat{\ell} /2^{m+1} \rfloor )  \leq \sum_{m=0}^{\log \hat{\ell}} 2^m ( \hat{\ell}/2^m)  =\hat{\ell} \log \hat{\ell}  $.

    We now show that $ \min\{t_{\max}, \texttt{cap}_{i+t_{\max}}  \}= \min\{t_{\max}, \texttt{cap}_i \} $ as follows.
    
    If $\texttt{cap}_i \geq t_{\max} $, then $ t_{\max} $ divides $\texttt{cap}_i$ and $\texttt{cap}_{i+t_{\max}}$, since both $\texttt{cap}_i$ and $t_{\max} $ are powers of 2. Hence, $ \min\{t_{\max}, \texttt{cap}_i \} = t_{\max} =   \min\{t_{\max}, \texttt{cap}_{i+t_{\max}}  \}$.
    
    Otherwise, $\texttt{cap}_i < t_{\max} $. Then $\texttt{cap}_i = \texttt{cap}_{i+t_{\max}}$, since $\texttt{cap}_i$ and $t_{\max} $ are powers of 2. Thus, $\min\{t_{\max}, \texttt{cap}_i \} =  \min\{t_{\max}, \texttt{cap}_{i+t_{\max}}  \}  $.

    Thus, we get $\sum_{i=0}^{\ell_{\max}} \min \{ t_{\max} , 
    \texttt{cap}_i \}  \leq  
    \frac{\ell_{\max}}{t_{\max}}  
    \sum_{i=0}^{t_{\max}} \min \{t_{\max}, \texttt{cap}_i \}  
   \allowbreak = O(\ell_{\max} \cdot \log t_{\max}) $.
    This is because, $\sum_{i=0}^{t_{\max}} \texttt{cap}_i  = O(t_{\max} \cdot \log t_{\max})$, as we proved in the beginning of the proof of this lemma.

By \autoref{SummingUp}, $T \leq 2\sum_{j=1}^d \work_j(V)/r  + \sum_{i=0}^{\ell_{\max}} \min\{t_{\max},  \texttt{cap}_i\}$ and we just showed $\sum_{i=0}^{\ell_{\max}} \min \{ t_{\max} , 
    \texttt{cap}_i \} = O(\ell_{\max} \cdot \log t_{\max}) $. Combining these, the lemma is proved.

\end{proof}

\paragraph{\textbf{ Makespan Lower Bound of the Optimal Algorithm.}}
\begin{lemma}
\label{lem:opt-lower}
Let $V$ be the set of vertices (jobs) in the input DAG with total resource capacity $r$ (normalized for each resource type separately) for each resource $\mathcal{R}_i$ where $i\in[1,d]$. Let $\ell_{\max}$ be the maximum level a job is assigned by algorithm \onl. Then, the lower bound of makespan $T_{\opt}$ of the optimal algorithm \opt is as follows. \hao{You actually want $ \max_{v \in G} \phi(v) +t_v $ (-1 if you assigned the lowest level to be 1)}
$$
T_{\opt} \geq \max\Bigg(\frac{\sum_{j=1}^d \work_j(V)}{d \cdot r}, \frac{\ell_{\max}}{2}\Bigg).
$$
\end{lemma}
\begin{proof}

At any time-step $t$ during the executing of \opt, it can use all $r$ units of resource $\mathcal{R}_j$. Hence, time-step $t$ contributes at most $r$ in the quantity $\work_j(V)$. Summing over all $d$ resource types, time-step $t$ can contribute at most $d \cdot r$ in the quantity $\sum_{j=1}^d \work_j(V)$. Hence,
$$
T_{\onl} \ge \frac{\sum_{j=1}^d \work_j(V)}{d \cdot r}.
$$

Let $z$ be the sink node of the input DAG. Consider the source-to-sink path of longest length (also equals to the depth $\delta$ of the sink node) in the input DAG, where the length of a source-to-sink path is the sum of processing times of jobs/nodes in the path. Hence, $T_{\opt} \ge \delta(z)$

The sink node gets the maximum level assignment $\ell_{\max}$.
From \hao{added lemma number} Lemma~\ref{lvlleq2depth}, we know that the level of a node is at most $2\cdot \delta(z)$. Hence,  
$$
T_{\opt} \ge \frac{\ell_{\max}}{2}.
$$

Combining these two lower bounds, the lemma is proved.

\end{proof}

\paragraph{\textbf{ Proof of Theorem~\ref{thm:onl-upper}.}}
\begin{proof}
Comparing the makespan upper bound of our online algorithm (from Lemma~\ref{bdInTermstmax} and Lemma~\ref{lower1}) with OPT, we get that our online algorithm is $O(d +  \min\{\log t_{\max}, \log |V| \})$-competitive. The main theorem of this section is thus proved.
\end{proof}

\section{Greedy Algorithms can be $\Omega(n)$-Competitive.}
In this section, we show that greedy algorithms can have poor competitive ratio. In particular, we show the following.

\begin{lemma}
\label{greedy-lower}
There are problem instances of $n$ jobs such that any greedy algorithm is $\Omega(n)$-competitive.
\end{lemma}

\begin{proof}
    
\begin{figure}
    \centering
\includegraphics[width=0.7\textwidth]{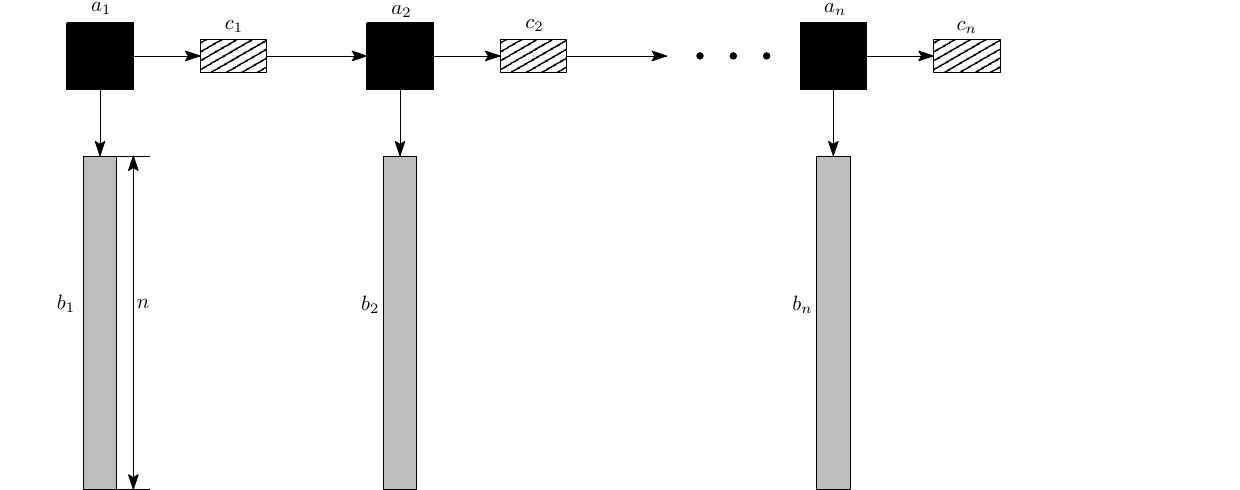}
   \caption{Lower bound construction of greedy algorithms}
    \label{fig:greedy}
\end{figure}
In the DAG, there are three types of jobs: Job $a_i$ takes 1 unit of resource and 1 unit of time. Job $b_i$ takes $\epsilon = 1/2n$ units of resource and $n$ units of time. Job $c_i$ takes $\epsilon = 1/2n$ units of resource and $1$ units of time. The precedence relations are as follows: for each $i$, $a_i \prec b_i$ and $a_i \prec c_i$. Moreover, $c_i \prec a_{i+1}$ (see Figure~\ref{fig:greedy}.
Let the resource budget be 1.

Any online greedy algorithm runs jobs $b_i$ and $c_i$ in parallel. Since job $a_{i+1}$ takes all the resource budget, until job $b_i$ is finished, job $a_{i+1}$ can not start. However, there is no precedence constraint between job $b_i$ and $a_{i+1}$. Thus, each job $b_i$ creates a bottleneck of size $n$. Since there are $n$ number of $b_i$ jobs, any greedy algorithm's makespan is $\Omega(n^2)$.

On the other hand, an optimal schedule would be as follows: Execute all the $a_i$ and $c_i$ jobs, but not $b_i$. This takes $O(n)$ time. After that, execute all the $b_i$ jobs in parallel, which takes another $n$ time steps. Hence, the makespan is $O(n)$. This proves the lemma.
\end{proof}

\section{Conclusions and Open Problems}
We resolve an important open question in resource scheduling: Can a polynomial-time $O(1)$-factor approximation algorithm be found? 
We also show a connection between the scheduling problem and a seemingly unrelated problem called the shortest common super-sequence (SCS) problem and prove that improving the current state-of-the-art in the scheduling problem would make significant progress in the sequencing problem. We then consider the online setting of the resource scheduling problem and present a deterministic online algorithm with matching upper and lower bounds in the competitive ratio.

We introduce a new toolbox called ``chains" to prove all the lower bounds. We hope this tool will be further helpful in proving lower bounds in other resource scheduling problems.

An immediate open question is to investigate the other direction in the relationship of scheduling and sequencing: If we improve the approximation ratio for SCS, would it imply improving the same for resource scheduling?
Another direction for future research is to investigate better resource scheduling algorithms for other important objectives, such as minimizing weighted completion time, weighted flow time, or $\ell_{p}$-norms.
\subsection*{Acknowledgments}
We thank the anonymous reviewers of SPAA for their helpful suggestions.
\appendix
\newpage
\section{Appendix}
\label{sec:appendix}
\paragraph{Proof of Lemma~\ref{lvlleq2depth}}
\begin{proof}
We prove the stronger statement that $\delta(v) \geq \frac{\psi(v)- t_v}{2} $ by induction on  the number of predecessors of $v$. 
Let $s$ be the source.

  \textbf{Base case:}  The inequality holds since $\delta(s) = \psi(s) =1$ and $t_s =0$.
  
  \textbf{Induction step:}
  Let $v$ be a  node that is not the source.
  Let $w$ be an immediate predecessor of $v$ maximizing $\delta(w) + t_w$.
   By induction, $\delta(w) \geq \frac{\psi(w)- t_w}{2} $.  
   The smallest multiple of $t_v$ greater than $\psi(w) + t_w$ is at most $\psi(w) + t_w+ t_v$.  
   So $\psi(v) \leq \psi(w) + t_w+ t_v$.
Then   $ \frac{ \psi(v) -t_v}{2}  \leq  \frac{ \psi(w) +t_w + t_v  -t_v}{2} \leq  \delta(w) + t_w \leq \delta(v)   $.
\end{proof}

\paragraph{Proof of Property \ref{parallelmakespan} }
Let us first prove the following lemma
\begin{lemma}\label{countingspecial}
    Let $ \{i_1,i_2,\ldots, i_p \}$ be a set of $p$ distinct positive integers, and let  $\{m_1,m_2,.., m_p\}$ be another set of $p$ positive integers.
    Let $q_0,q_1,.., q_r$ be a sequence of powers of 2.
    Define  $ s_i= | \{ j\in [r] : q_j  \geq 2^i   \} |   $.
    Suppose for each $t=1,2, \ldots ,p$, $s_{ i_t}  \geq 2^{m_t}  /2^{i_t} $.
    Then $\sum_{j=0}^r  q_j \geq \Omega(2^{m_1} + 2^{m_2} + 2^{m_3} + \cdots + 2^{m_p})$
\end{lemma}
\begin{proof}
    W.l.o.g., assume $ 0<i_1 <i_2 <..< i_p $. 
    For each $t=1,2,..,p$, we define a bucket which corresponds to $m_t$ and $i_t$.
    Each $ q_j $ allocates $ 2^{i_t} $ units of credit to each bucket $t$ such that $2^{i_t}  \leq q_j $.
    The total amount of credit allocated by $q_j$ is at most $\sum_{ t: \ \ 2^{i_t}  \leq q_j } 2^{i_t}  \leq \sum_{ t=0 }^{\log_2 (q_j)}  2^t \leq 2 q_j  $. 
    The total amount of credit allocated to bucket $t$ is $ 2^{i_t} s_{i_t}  \geq 2^{m_t} $.
    Thus the total amount of credit allocated to all buckets is $2^{m_1} + 2^{m_2} + 2^{m_3} + \cdots + 2^{m_p}$. 
    This shows that $\sum_{j=0}^r  q_j \geq \Omega(2^{m_1} + 2^{m_2} + 2^{m_3} + \cdots + 2^{m_p})$.
\end{proof}

\begin{proof}(of Property~\ref{parallelmakespan})
    Let $\mathcal{A}$ be any algorithm.  By \autoref{immediatestartchain}, we can modify $\mathcal{A}$ so that it does not start any job   while another job has already started but not finished.
     Henceforth, we assume that $\mathcal{A}$ does not start any job while another job is running.

    Let $ \mathcal{A}_{job}(t) $ denote the length of the longest job started or running at time $t$. 
     Define $ q_0=  \mathcal{A}_{job}(0) $ and $ q_{i+1}= \mathcal{A}_{job}( \sum_{j=0}^i q_j) $.
     Note that at each time $\sum_{j=0}^i q_j$, a job of length $q_{i+1}$ is started.
     Let $r$ be the index of the last $q_i$.
    Since we do not start any jobs while another is running, each job is started at time  $ \sum_{j=0}^i q_j$ for some $i$.
    Further if the length of the job started is $ q' $, than $q' \leq q_j $.
    Thus if chain $C(m_t,i_t)$ is completed, then there must be indices  $ j(1),j(2),.., j( 2^{m_t} /2^{i_t} ) $ such that the $i$-th skinny job of $C(m_1,i_1)$ is started at time $\sum_{j=0}^{j(i)-1} q_j$. 
    Note $ q_{j(i)}  \geq 2^{i_t} $.
    Thus, for $\mathcal{A}$ to complete chain $C(m_t,i_t)$, there must be at least $ 2^{m_t} /2^{t_1} $ indices  $ j(1),j(2),.., j( 2^{m_t} /2^{i_t} )  $ such that $q_{j(i)} \geq 2^{i_1} $.
    By \autoref{countingspecial}, the makespan of $\mathcal{A}$, or  $\sum_{j=0}^r  q_j $ is at least $ \Omega(2^{m_1} + 2^{m_2} + 2^{m_3} + \cdots + 2^{m_p})$.
    
\end{proof}
\newpage 

\bibliography{main}

\begin{thebibliography}{10}

\bibitem{AgrawalBD20}
Kunal Agrawal, Michael Bender, Rathish Das, William Kuszmaul, Enoch Peserico, and Michele Scquizzato.
\newblock Brief announcement: Green paging and parallel paging.
\newblock In {\em Proc. 32st ACM on Symposium on Parallelism in Algorithms and Architectures (SPAA)}, 2020.

\bibitem{AgrawalBD21}
Kunal Agrawal, Michael Bender, Rathish Das, William Kuszmaul, Enoch Peserico, and Michele Scquizzato.
\newblock Tight bounds of parallel paging and green paging.
\newblock In {\em Proceedings of the Fifteenth Annual ACM-SIAM Symposium on Discrete Algorithms (SODA)}, 2021.

\bibitem{DBLP:conf/spaa/AgrawalBDKPS22}
Kunal Agrawal, Michael~A. Bender, Rathish Das, William Kuszmaul, Enoch Peserico, and Michele Scquizzato.
\newblock Online parallel paging with optimal makespan.
\newblock In Kunal Agrawal and I{-}Ting~Angelina Lee, editors, {\em Proc. 34th ACM on Symposium on Parallelism in Algorithms and Architectures (SPAA), Philadelphia, PA, USA, July 11 - 14, 2022}, pages 205--216.

\bibitem{arunarani2019task}
AR~Arunarani, Dhanabalachandran Manjula, and Vijayan Sugumaran.
\newblock Task scheduling techniques in cloud computing: A literature survey.
\newblock {\em Future Generation Computer Systems}, 91:407--415, 2019.

\bibitem{augustine2006strip}
John Augustine, Sudarshan Banerjee, and Sandy Irani.
\newblock Strip packing with precedence constraints and strip packing with release times.
\newblock In {\em Proceedings of the eighteenth annual ACM symposium on Parallelism in algorithms and architectures}, pages 180--189, 2006.

\bibitem{azar2002line}
Yossi Azar and Leah Epstein.
\newblock On-line scheduling with precedence constraints.
\newblock {\em Discrete Applied Mathematics}, 119(1-2):169--180, 2002.

\bibitem{baker1983shelf}
Brenda~S Baker and Jerald~S Schwarz.
\newblock Shelf algorithms for two-dimensional packing problems.
\newblock {\em SIAM Journal on Computing}, 12(3):508--525, 1983.

\bibitem{banerjee1990MallableApx}
KPBP Banerjee.
\newblock An approximate algorithm for the partitionable independent task scheduling problem.
\newblock {\em Urbana}, 51:61801, 1990.

\bibitem{SCSExperimental}
Paolo Barone, Paola Bonizzoni, Gianluca~Delta Vedova, and Giancarlo Mauri.
\newblock An approximation algorithm for the shortest common supersequence problem: an experimental analysis.
\newblock In {\em Proceedings of the 2001 ACM Symposium on Applied Computing}, SAC '01, page 56–60, New York, NY, USA, 2001. Association for Computing Machinery.

\bibitem{bender2000scheduling}
Michael~A Bender and Michael~O Rabin.
\newblock Scheduling cilk multithreaded parallel programs on processors of different speeds.
\newblock In {\em Proceedings of the twelfth annual ACM symposium on Parallel algorithms and architectures}, pages 13--21, 2000.

\bibitem{berger2000hoard}
Emery~D Berger, Kathryn~S McKinley, Robert~D Blumofe, and Paul~R Wilson.
\newblock Hoard: A scalable memory allocator for multithreaded applications.
\newblock {\em ACM Sigplan Notices}, 35(11):117--128, 2000.

\bibitem{LTSnotAPX}
R.~Bhatia, S.~Khuller, and J.~Naor.
\newblock The loading time scheduling problem.
\newblock In {\em Proceedings of IEEE 36th Annual Foundations of Computer Science}, pages 72--81, 1995.

\bibitem{blumofe1995cilk}
Robert~D Blumofe, Christopher~F Joerg, Bradley~C Kuszmaul, Charles~E Leiserson, Keith~H Randall, and Yuli Zhou.
\newblock Cilk: An efficient multithreaded runtime system.
\newblock {\em ACM SigPlan Notices}, 30(8):207--216, 1995.

\bibitem{blumofe1993space}
Robert~D Blumofe and Charles~E Leiserson.
\newblock Space-efficient scheduling of multithreaded computations.
\newblock In {\em Proceedings of the twenty-fifth annual ACM symposium on Theory of computing}, pages 362--371, 1993.

\bibitem{blumofe1999scheduling}
Robert~D Blumofe and Charles~E Leiserson.
\newblock Scheduling multithreaded computations by work stealing.
\newblock {\em Journal of the ACM (JACM)}, 46(5):720--748, 1999.

\bibitem{BRUCKER19993}
Peter Brucker, Andreas Drexl, Rolf Möhring, Klaus Neumann, and Erwin Pesch.
\newblock Resource-constrained project scheduling: Notation, classification, models, and methods.
\newblock {\em European Journal of Operational Research}, 112(1):3--41, 1999.

\bibitem{chekuri2004multidimensional}
Chandra Chekuri and Sanjeev Khanna.
\newblock On multidimensional packing problems.
\newblock {\em SIAM journal on computing}, 33(4):837--851, 2004.

\bibitem{Malleable2013}
Chi-Yeh Chen and Chih-Ping Chu.
\newblock A 3.42-approximation algorithm for scheduling malleable tasks under precedence constraints.
\newblock {\em IEEE Transactions on Parallel and Distributed Systems}, 24(8):1479--1488, 2013.

\bibitem{christensen2017approximation}
Henrik~I Christensen, Arindam Khan, Sebastian Pokutta, and Prasad Tetali.
\newblock Approximation and online algorithms for multidimensional bin packing: A survey.
\newblock {\em Computer Science Review}, 24:63--79, 2017.

\bibitem{cormen2022introduction}
Thomas~H Cormen, Charles~E Leiserson, Ronald~L Rivest, and Clifford Stein.
\newblock {\em Introduction to algorithms}.
\newblock MIT press, 2022.

\bibitem{csirik1997shelf}
J{\'a}nos Csirik and Gerhard~J Woeginger.
\newblock Shelf algorithms for on-line strip packing.
\newblock {\em Information Processing Letters}, 63(4):171--175, 1997.

\bibitem{DasAB20}
Rathish Das, Kunal Agrawal, Michael Bender, Jonathan Berry, Benjamin Moseley, and Cynthia Phillips.
\newblock How to manage high-bandwidth memory automatically.
\newblock In {\em Proc. 32st ACM on Symposium on Parallelism in Algorithms and Architectures (SPAA)}, 2020.

\bibitem{das2025approximation}
Rathish Das and Hao Sun.
\newblock Approximation hardness of resource scheduling.
\newblock In {\em Proceedings of the 37th ACM Symposium on Parallelism in Algorithms and Architectures}, pages 46--61, 2025.

\bibitem{DasTsDu19}
Rathish Das, Shih-Yu Tsai, Sharmila Duppala, Jayson Lynch, Esther~M Arkin, Rezaul Chowdhury, Joseph~SB Mitchell, and Steven Skiena.
\newblock Data races and the discrete resource-time tradeoff problem with resource reuse over paths.
\newblock In {\em Proc. 31st ACM on Symposium on Parallelism in Algorithms and Architectures (SPAA)}, pages 359--368, 2019.

\bibitem{DBLP:conf/spaa/DeLayoZABBDMP22}
Daniel DeLayo, Kenny Zhang, Kunal Agrawal, Michael~A. Bender, Jonathan~W. Berry, Rathish Das, Benjamin Moseley, and Cynthia~A. Phillips.
\newblock Automatic {HBM} management: Models and algorithms.
\newblock In Kunal Agrawal and I{-}Ting~Angelina Lee, editors, {\em {SPAA} '22: 34th {ACM} Symposium on Parallelism in Algorithms and Architectures, Philadelphia, PA, USA, July 11 - 14, 2022}, pages 147--159.

\bibitem{demirci2018approximation}
G{\"o}kalp Demirci, Henry Hoffmann, and David~HK Kim.
\newblock Approximation algorithms for scheduling with resource and precedence constraints.
\newblock In {\em 35th Symposium on Theoretical Aspects of Computer Science (STACS 2018)}. Schloss Dagstuhl-Leibniz-Zentrum fuer Informatik, 2018.

\bibitem{citation-key}
Huseyin~Gokalp Demirci.
\newblock Approximation algorithms for capacitated k-median and scheduling with resource and precedence constraints.

\bibitem{diab2013dynamic}
Khaled~M Diab, M~Mustafa Rafique, and Mohamed Hefeeda.
\newblock Dynamic sharing of gpus in cloud systems.
\newblock In {\em 2013 IEEE International Symposium on Parallel \& Distributed Processing, Workshops and Phd Forum}, pages 947--954. IEEE, 2013.

\bibitem{Leung1989Mallable}
Jianzhong Du and Joseph Y.-T. Leung.
\newblock Complexity of scheduling parallel task systems.
\newblock {\em SIAM Journal on Discrete Mathematics}, 2(4):473--487, 1989.

\bibitem{epstein1998lower}
Leah Epstein.
\newblock Lower bounds for on-line scheduling with precedence constraints on identical machines.
\newblock In {\em Approximation Algorithms for Combinatiorial Optimization: International Workshop APPROX'98 Aalborg, Denmark, July 18--19, 1998 Proceedings 1}, pages 89--98. Springer, 1998.

\bibitem{epstein2000note}
Leah Epstein.
\newblock A note on on-line scheduling with precedence constraints on identical machines.
\newblock {\em Information processing letters}, 76(4-6):149--153, 2000.

\bibitem{DBLP:conf/stoc/FeldmannKST93}
Anja Feldmann, Ming{-}Yang Kao, Jir{\'{\i}} Sgall, and Shang{-}Hua Teng.
\newblock Optimal online scheduling of parallel jobs with dependencies.
\newblock In S.~Rao Kosaraju, David~S. Johnson, and Alok Aggarwal, editors, {\em Proceedings of the Twenty-Fifth Annual {ACM} Symposium on Theory of Computing, May 16-18, 1993, San Diego, CA, {USA}}, pages 642--651. {ACM}, 1993.

\bibitem{frigo1998implementation}
Matteo Frigo, Charles~E Leiserson, and Keith~H Randall.
\newblock The implementation of the cilk-5 multithreaded language.
\newblock In {\em Proceedings of the ACM SIGPLAN 1998 conference on Programming language design and implementation}, pages 212--223, 1998.

\bibitem{garey1975bounds}
Michael~R Garey and Ronald~L. Graham.
\newblock Bounds for multiprocessor scheduling with resource constraints.
\newblock {\em SIAM Journal on Computing}, 4(2):187--200, 1975.

\bibitem{SCS2(3)APX}
Zvi Gotthilf and Moshe Lewenstein.
\newblock Improved approximation results on the shortest common supersequence problem.
\newblock In Jussi Karlgren, Jorma Tarhio, and Heikki Hyyr{\"o}, editors, {\em String Processing and Information Retrieval}, pages 277--284, Berlin, Heidelberg, 2009. Springer Berlin Heidelberg.

\bibitem{graham1966bounds}
Ronald~L Graham.
\newblock Bounds for certain multiprocessing anomalies.
\newblock {\em Bell system technical journal}, 45(9):1563--1581, 1966.

\bibitem{hall1997scheduling}
Leslie~A Hall, Andreas~S Schulz, David~B Shmoys, and Joel Wein.
\newblock Scheduling to minimize average completion time: Off-line and on-line approximation algorithms.
\newblock {\em Mathematics of operations research}, 22(3):513--544, 1997.

\bibitem{han2007strip}
Xin Han, Kazuo Iwama, Deshi Ye, and Guochuan Zhang.
\newblock Strip packing vs. bin packing.
\newblock In {\em Algorithmic Aspects in Information and Management: Third International Conference, AAIM 2007, Portland, OR, USA, June 6-8, 2007. Proceedings 3}, pages 358--367. Springer, 2007.

\bibitem{hurink2008online}
Johann~L Hurink and Jacob~Jan Paulus.
\newblock Online algorithm for parallel job scheduling and strip packing.
\newblock In {\em Approximation and Online Algorithms: 5th International Workshop, WAOA 2007, Eilat, Israel, October 11-12, 2007. Revised Papers 5}, pages 67--74. Springer, 2008.

\bibitem{im2015tight}
Sungjin Im, Nathaniel Kell, Janardhan Kulkarni, and Debmalya Panigrahi.
\newblock Tight bounds for online vector scheduling.
\newblock In {\em 2015 IEEE 56th Annual Symposium on Foundations of Computer Science}, pages 525--544. IEEE, 2015.

\bibitem{jansen2002Malleble}
Jansen and Porkolab.
\newblock Linear-time approximation schemes for scheduling malleable parallel tasks.
\newblock {\em Algorithmica}, 32:507--520, 2002.

\bibitem{jansen2004Malleble}
Klaus Jansen.
\newblock Scheduling malleable parallel tasks: An asymptotic fully polynomial time approximation scheme.
\newblock {\em Algorithmica}, 39:59--81, 2004.

\bibitem{jansen2019approximation}
Klaus Jansen, Marten Maack, and Malin Rau.
\newblock Approximation schemes for machine scheduling with resource (in-) dependent processing times.
\newblock {\em ACM Transactions on Algorithms (TALG)}, 15(3):1--28, 2019.

\bibitem{Malleable2006}
Klaus Jansen and Hu~Zhang.
\newblock An approximation algorithm for scheduling malleable tasks under general precedence constraints.
\newblock {\em ACM Trans. Algorithms}, 2(3):416–434, jul 2006.

\bibitem{supseqnotAPX}
Tao Jiang and Ming Li.
\newblock On the approximation of shortest common supersequences and longest common subsequences.
\newblock {\em SIAM Journal on Computing}, 24(5):1122--1139, 1995.

\bibitem{precedenceprocessor}
Safia Kedad-Sidhoum, Florence Monna, and Denis Trystram.
\newblock Scheduling tasks with precedence constraints on hybrid multi-core machines.
\newblock In {\em 2015 IEEE International Parallel and Distributed Processing Symposium Workshop}, pages 27--33, 2015.

\bibitem{KOLISCH2001249}
R~Kolisch and R~Padman.
\newblock An integrated survey of deterministic project scheduling.
\newblock {\em Omega}, 29(3):249--272, 2001.

\bibitem{DBLP:conf/icml/LassotaLMS23}
Alexandra~Anna Lassota, Alexander Lindermayr, Nicole Megow, and Jens Schl{\"{o}}ter.
\newblock Minimalistic predictions to schedule jobs with online precedence constraints.
\newblock In Andreas Krause, Emma Brunskill, Kyunghyun Cho, Barbara Engelhardt, Sivan Sabato, and Jonathan Scarlett, editors, {\em International Conference on Machine Learning, {ICML} 2023, 23-29 July 2023, Honolulu, Hawaii, {USA}}, volume 202 of {\em Proceedings of Machine Learning Research}, pages 18563--18583. {PMLR}, 2023.

\bibitem{ScheduleNPhard}
J.~K. Lenstra and A.~H. G.~Rinnooy Kan.
\newblock Complexity of scheduling under precedence constraints.
\newblock {\em Operations Research}, 26(1):22--35, 1978.

\bibitem{MalleableTree}
Renaud Lepère, Grégory Mounié, and Denis Trystram.
\newblock An approximation algorithm for scheduling trees of malleable tasks.
\newblock {\em European Journal of Operational Research}, 142(2):242--249, 2002.

\bibitem{Ludwig1994SchedulingMallable}
Walter Ludwig and Prasoon Tiwari.
\newblock Scheduling malleable and nonmalleable parallel tasks.
\newblock In {\em ACM-SIAM Symposium on Discrete Algorithms}, 1994.

\bibitem{mastrolilli2008acyclic}
Monaldo Mastrolilli and Ola Svensson.
\newblock (acyclic) job shops are hard to approximate.
\newblock In {\em 2008 49th Annual IEEE Symposium on Foundations of Computer Science}, pages 583--592. IEEE, 2008.

\bibitem{megow2009stochastic}
Nicole Megow and Tjark Vredeveld.
\newblock Stochastic online scheduling with precedence constraints.
\newblock 2009.

\bibitem{motwani1995randomized}
Rajeev Motwani and Prabhakar Raghavan.
\newblock {\em Randomized algorithms}.
\newblock Cambridge university press, 1995.

\bibitem{Mounie1999Malleble}
Gregory Mounie, Christophe Rapine, and Dennis Trystram.
\newblock Efficient approximation algorithms for scheduling malleable tasks.
\newblock In {\em Proceedings of the Eleventh Annual ACM Symposium on Parallel Algorithms and Architectures}, SPAA '99, page 23–32, New York, NY, USA, 1999. Association for Computing Machinery.

\bibitem{niemeier2015scheduling}
Martin Niemeier and Andreas Wiese.
\newblock Scheduling with an orthogonal resource constraint.
\newblock {\em Algorithmica}, 71:837--858, 2015.

\bibitem{binSCSapx}
M.~Duella P.~Bonizzoni and G.~Mauri.
\newblock Approximation complexity of longest common subsequence and shortest common supersequence over fixed alphabet.
\newblock {\em Technical Report 117/94, Università degli Studi Milano, Italy}.

\bibitem{perotin2021moldable}
Lucas Perotin, Hongyang Sun, and Padma Raghavan.
\newblock Multi-resource list scheduling of moldable parallel jobs under precedence constraints.
\newblock In {\em Proceedings of the 50th International Conference on Parallel Processing}, pages 1--10, 2021.

\bibitem{perotin2025new}
Lucas Perotin, Hongyang Sun, and Padma Raghavan.
\newblock A new algorithm for online scheduling of rigid task graphs with near-optimal competitive ratio.
\newblock In {\em Proceedings of the 37th ACM Symposium on Parallelism in Algorithms and Architectures}, pages 210--224, 2025.

\bibitem{queyranne2006approximation}
Maurice Queyranne and Andreas~S Schulz.
\newblock Approximation bounds for a general class of precedence constrained parallel machine scheduling problems.
\newblock {\em SIAM Journal on Computing}, 35(5):1241--1253, 2006.

\bibitem{SCSBinNPHard}
Kari{-}Jouko R{\"{a}}ih{\"{a}} and Esko Ukkonen.
\newblock The shortest common supersequence problem over binary alphabet is np-complete.
\newblock {\em Theor. Comput. Sci.}, 16:187--198, 1981.

\bibitem{schneider2006scalable}
Scott Schneider, Christos~D Antonopoulos, and Dimitrios~S Nikolopoulos.
\newblock Scalable locality-conscious multithreaded memory allocation.
\newblock In {\em Proceedings of the 5th international symposium on Memory management}, pages 84--94, 2006.

\bibitem{svensson2010conditional}
Ola Svensson.
\newblock Conditional hardness of precedence constrained scheduling on identical machines.
\newblock In {\em Proceedings of the forty-second ACM symposium on Theory of computing}, pages 745--754, 2010.

\bibitem{voorsluys2011introduction}
William Voorsluys, James Broberg, and Rajkumar Buyya.
\newblock Introduction to cloud computing.
\newblock {\em Cloud computing: Principles and paradigms}, pages 1--41, 2011.

\bibitem{ye2009note}
Deshi Ye, Xin Han, and Guochuan Zhang.
\newblock A note on online strip packing.
\newblock {\em Journal of Combinatorial Optimization}, 17(4):417--423, 2009.

\bibitem{zhu2011multimedia}
Wenwu Zhu, Chong Luo, Jianfeng Wang, and Shipeng Li.
\newblock Multimedia cloud computing.
\newblock {\em IEEE Signal Processing Magazine}, 28(3):59--69, 2011.

\end{thebibliography}
\end{document}